\documentclass[a4paper,11pt]{article}

\usepackage{cmap}     
\usepackage{amsmath,amssymb}
\usepackage[utf8]{inputenc}
\usepackage[T2A]{fontenc}  
\usepackage{graphicx}    
\usepackage[margin=1in]{geometry}
\usepackage{fancyhdr}    
\usepackage{setspace}
\usepackage{wrapfig}
\usepackage{subfig}
\usepackage{listings}
\usepackage{color}
\usepackage{setspace}
\usepackage{textcomp}
\usepackage{float}
\usepackage{listings}
\usepackage{amsthm}
\usepackage{algorithm}
\usepackage{algpseudocode}
\usepackage{afterpage}
\usepackage[toc,page]{appendix}
\usepackage{natbib}
\usepackage{bbm}
\usepackage{footnote}

\makesavenoteenv{tabular}

\newtheorem{theorem}{Theorem}
\newtheorem{lemma}{Lemma}

\DeclareMathOperator{\cov}{cov}

\begin{document}    
\title{Numerical analysis of an extended structural default model with mutual liabilities and jump risk} 
\author{Vadim Kaushansky\thanks{The first author gratefully acknowledges support from the Economic and Social Research Council and Bank of America Merrill Lynch}
\footnote{Mathematical Institute \& Oxford-Man Institute, University of Oxford, UK, E-mail: vadim.kaushansky@maths.ox.ac.uk},  Alexander Lipton\footnote{Massachusetts Institute of Technology, Connection Science, Cambridge, MA, USA, E-mail: alexlipt@mit.edu}, Christoph Reisinger\footnote{Mathematical Institute  \& Oxford-Man Institute, University of Oxford, UK, E-mail: christoph.reisinger@maths.ox.ac.uk}}    
\date{}
   
\maketitle 
\begin{abstract}
We consider a structural default model in an interconnected banking network as in \cite{Lipton2015}, with mutual obligations between each pair of banks. We analyse the model numerically for two banks with jumps in their asset value processes. Specifically, we develop a finite difference method for the resulting two-dimensional partial integro-differential equation, and study its stability and consistency. We then compute joint and marginal survival probabilities, as well as prices of credit default swaps (CDS), first-to-default swaps (FTD), credit and debt value adjustments (CVA and DVA). Finally, we calibrate the model to market data and assess the impact of jump risk.
\end{abstract}

\medskip

%\noindent
%{\bf Highlights:}
%%The novel results of this paper are as follows:
%\begin{itemize}
%\item We analyze a two-dimensional structural default model with interbank liabilities and negative exponential jumps; in particular, we calibrate the model to the market and analyze the impact of jumps on joint and marginal survival probabilities; 
%\item we develop a new finite-difference method to solve the multidimensional PIDE, which is of second order consistent in both time and space variables; 
%\item we prove the von Neumann and $l_2$ stability of the method; % extending \cite{intHoutStability} from PDEs to PIDEs;
%%to our knowledge, this is the first result on stability of a splitting scheme for this type of multi-dimensional PIDE taking into account Dirichlet boundary conditions;
%\item we demonstrate empirically that in the presence of discontinuous terminal and boundary conditions, second order of convergence can be maintained by local averaging of the data and suitable refinement of the timestep close to maturity.
%\end{itemize}

\section{Introduction}

The estimation and mitigation of counterparty credit risk %and the estimation 
has become a pillar of financial risk management.
The impact of such risks on financial derivatives is explicitly acknowledged by a valuation adjustment.
For an exposition of the background and mathematical models we refer the reader to \cite{gregory2012counterparty}.
Although reduced-from models provide for a more direct simulation of default events and are commonly used in financial institutions, we follow here a structural approach which maps the capital structure of a bank into stochastic processes for equity and debt, and models default as the hitting of a lower barrier, as in \cite{BlackCox}.
An extensive literature review of further developments of this model is given in \cite{LiptonSepp}.

A particular concern to regulators and central banks is the impact of default of an entity on the financial system -- credit contagion. Of the various channels of such systemic risk described in \cite{hurd2016contagion}, we focus here on dependencies through asset correlation and interbank liabilities.
Specifically, we consider the extended structural default model introduced in \cite{Lipton2015}, where asset values are assumed to follow stochastic processes with correlated diffusion and jump components, and where mutual liabilities can lead to default cascades.

\cite{LiptonItkin2015} consider the model without jumps and obtain explicit expressions for several quantities of interest including the joint and marginal survival probabilities as well as CDSs and FTD prices.
They demonstrate that mutual liabilities can have a large impact on the survival probabilities of banks. Thus, a shock of one bank can cause ripples through the whole banking system. 

We focus here on the numerical computation of survival probabilities and credit products in the extended model with jumps, where closed-from expressions are no longer available.
Our work is therefore most similar to \cite{LiptonItkin2014}, who 
develop a finite difference method for the resulting partial integro-differential equation (PIDE) where the integral term results from a fairly general correlated Levy process in the jump component. % using a special splitting scheme.
By Strang splitting into the diffusion and jump operators, overall second order consistency in the timestep is obtained.
% even though these operators due not commute.
Hereby, the multi-dimensional diffusion operator is itself split dimensionwise using the Hundsdorfer-Verwer (HV) scheme, and the jump operator is 
treated as a pseudo-differential operator, which allows efficient evaluation of the discretised operator by an iterative procedure.
Stability of each of the steps is guaranteed under standard conditions.

Our approach is more straightforward in that we apply a modification of the HV scheme directly, where we treat the jump term in the same way as the cross-derivative term in the classical HV scheme.
For the analysis we consider infinite meshes, i.e.\ ignoring the boundaries,
%, where a zero Dirichlet boundary condition is applied at the default boundary, 
such that the discrete operators are also infinite-dimensional.
In our analysis we build heavily on the results in \cite{intHoutStability} on stability of the PDE with cross-derivatives (but no integral term) and periodic boundary conditions on a finite mesh.

We show that the (unconditional) von Neumann stability of the scheme is not materially affected by the jump operator, as its contribution to the symbol of the scheme is of a lower order in the mesh size.
For concreteness, we restrict ourselves to the model with negative exponential jumps described in \cite{Lipton2015}.
This allows a simple recursive computation of the discrete jump operator and gives an explicit form of its eigenvalues. However, the analysis can in principle be extended to other jump size distributions.

%However, as the discretised operators do not commute -- through the presence of absorbing boundary conditions -- the operators
%are not simultaneously diagonalisable and the problem remains that the eigenvalues of the combined splitting operator are not directly obtained from the eigenvalues of the individual operators. By showing and using that the commutator is of low rank, we are still able to establish that the operators are simultaneously triangularisable. We thereby obtain the eigenvalues of the combined operator and can show that the scheme is stable in the von Neumann sense.\footnote{Incidentally, this idea can also be applied to the schemes in  \cite{LiptonItkin2014}.}

A survey of splitting methods in finance is found in \cite{toivanen2015application}. These are roughly arranged in two groups: splitting by dimension (for multi-dimensional PDEs; such as in \cite{intHoutStability}), and splitting by operator type 
(for PIDEs, diffusion and jumps; such as in \cite{andersen2000jump}). 
\cite{LiptonItkin2014} perform these two splittings successively as described above.
To our knowledge, the present paper is the first to perform and analyse splitting into dimensions and jumps simultaneously.
%Even though this is a special case of the model considered in \cite{LiptonItkin2014}, it requires a special treatment that can be done in more sophisticated way. 

The scheme is constructed to be second order consistent with the continuous integro-differential operator applied to smooth functions. However, the discontinuities in the data lead to empirically observed convergence of only first order in both space and time step.
To address this, we apply a spatial smoothing technique discussed in  \cite{pooley2003} for discontinuous option payoffs in the one-dimensional setting, and a change of the time variable to square-root time (see \cite{reisinger2013}), equivalent to a quadratically refined time mesh close to maturity, in order to restore second order convergence.
We emphasise that the presence of %absorbing boundary conditions and 
discontinuous initial data is essential to the nature of P(I)DE models of credit risk. Hence this approach improves on previous works in a key aspect of the numerical solution.

Similar to \cite{LiptonItkin2015}, we restrict the analysis to the two-dimensional case, but there is no fundamental problem in extending the method to multiple dimensions. However, due to the curse of dimensionality, for more than three-dimensional problems, standard finite-difference methods are computationally too expensive. The two-dimensional case already allows us to investigate various important model characteristics, such as joint and marginal survival probabilities, prices of credit derivatives, Credit and Debt Value Adjustments, and specifically the impact of mutual obligations.

In this paper, we consider both unilateral and bilateral counterparty risk as discussed in \cite{LiptonSav}. For the unilateral case, the model with two banks is considered, where one is a reference name and the other is either a protection buyer or a protection seller, while for the bilateral case, reference name, protection seller, and protection buyer are considered together, which leads to a three-dimensional problem. We give the equations in the Appendix, but do not include computations.

Moreover, we provide a calibration of the model with negative exponential jumps to market data, and for this calibrated model assess the impact of mutual obligations on survival probabilities.

%can be calibrated much easier than for general Levy processes. 
%\noindent
%{\bf Highlights:}
The novel results of this paper are as follows:
\begin{itemize}
\item We analyze a two-dimensional structural default model with interbank liabilities and negative exponential jumps; in particular, we calibrate the model to the market and analyze the impact of jumps on joint and marginal survival probabilities; 
\item we develop a new finite-difference method to solve the multidimensional PIDE, which is second order consistent in both time and space variables; 
\item we prove the von Neumann and $l_2$ stability of the method; % extending \cite{intHoutStability} from PDEs to PIDEs;
%to our knowledge, this is the first result on stability of a splitting scheme for this type of multi-dimensional PIDE taking into account Dirichlet boundary conditions;
\item we demonstrate empirically that in the presence of discontinuous terminal and boundary conditions, second order of convergence can be maintained by local averaging of the data and suitable refinement of the timestep close to maturity.
\end{itemize}

The rest of the paper is organized as follows. In Section 2, we formulate the model for two banks with jumps, which is a simplified formulation of \cite{Lipton2015} for two banks only. Then, we briefly discuss how to compute various model characteristics. In Section 3 we propose a numerical scheme for a general pricing problem; we further prove its stability and consistency. In Section 4 we provide numerical results for various model characteristics computed with the numerical scheme from Section 3. In Section 5 we calibrate the model to the market, and in Section 6 we conclude.
\section{Model}

We consider %a special case of 
the model in \cite{Lipton2015} %for interconnected banking network 
for two banks. Assume that the banks have external assets and liabilities, $A_i$ and $L_i$ respectively, for $i = 1, 2$, and interbank mutual liabilities $L_{12}$ and $L_{21}$, where $L_{ij}$ is the amount the $i$-th bank owes to the $j$-th bank. Then, the total assets and liabilities for banks 1 and 2 are
\begin{equation}
	\begin{aligned}
		& \tilde{A}_1 = A_1 + L_{21}, \quad \tilde{L}_1 = L_1 + L_{12}, \\
		& \tilde{A}_2 = A_2 + L_{12}, \quad \tilde{L}_2 = L_2 + L_{21}.
	\end{aligned}
\end{equation}

\subsection{Dynamics of assets and liabilities}

As in \cite{Lipton2015}, we assume that the firms' asset values before default are governed by
\begin{equation}
	\label{assets_dynamics}
	\frac{d A_i}{A_i} = (\mu - \kappa_i \lambda_i(t)) \, d t + \sigma_i\, d W_i(t) + (e^{J_i} - 1) \, d N_i(t), \quad i = 1, 2,
\end{equation}
where $\mu$ is the deterministic growth rate, and, for $i=1,2$, $\sigma_i$ are the corresponding volatilities, $W_i$ are correlated standard Brownian motions,
\begin{equation}
	d W_1(t) d W_2(t) = \rho \, d t, 
\end{equation}
with correlation $\rho$, $N_i$ are Poisson processes independent of $W_i$, $\lambda_i$ are the intensities of jump arrivals, $J_i$ are random negative exponentially distributed jump sizes with probability density function
\begin{equation}
	\tilde{\omega}_i(s) = \left\{
	\begin{aligned}
		& 0, & s > 0, \\
		& \vartheta_i e^{\vartheta_i s}, & s \le 0,
	\end{aligned}
	\right.
\end{equation} 
with parameters $\vartheta_i > 0$, and $\kappa_i$ are jump compensators
\begin{equation}
	\kappa_i = \mathbb{E} [e^{J_i} - 1] = -\frac{1}{\varsigma_i + 1}.
\end{equation}
The jump processes are correlated in the spirit of \cite{MarshallOlkin}. Consider independent Poisson processes $N_{\{1\}}(t), N_{\{2\}}(t)$ and $N_{\{12\}}(t)$, with the corresponding intensities $\lambda_{\{1\}}, \lambda_{\{2\}}$ and $\lambda_{\{12\}}$. Then, we define the processes $N_1(t)$ and $N_2(t)$ as
\begin{equation}
	\begin{aligned}
		N_i(t) &= N_{\{i\}}(t) + N_{\{12\}}(t), \qquad i=1,2,\\
		\lambda_i &= \lambda_{\{i\}} + \lambda_{\{12\}},
	\end{aligned}
\end{equation} 
%\begin{equation}
%	\begin{aligned}
%		& N_1(t) = N_{\{1\}}(t) + N_{\{12\}}(t), \\
%		& N_2(t) = N_{\{2\}}(t) + N_{\{12\}}(t), \\	
%	\end{aligned}
%\end{equation}
%with
%\begin{equation}
%	\begin{aligned}
%		& \lambda_1 = \lambda_{\{1\}} + \lambda_{\{12\}}, \\
%		& \lambda_2 = \lambda_{\{2\}} + \lambda_{\{12\}}. \\	
%	\end{aligned}
%\end{equation}
%Using this jump definition, we assume that 
i.e., there are both systemic and idiosyncratic sources of jumps.

We assume that the liabilities are deterministic and have the following dynamics
\begin{equation}
	\frac{d L_i}{L_i} = \mu \, d t, \qquad \frac{d L_{ij}}{L_{ij}} = \mu \, d t,
\end{equation}
where $\mu$ is the same growth rate as defined in \eqref{assets_dynamics}. For pricing purposes, under the risk-neutral measure, we consider $\mu$ as a risk-free short rate. In the following, we take for simplicity $\mu = 0$, but the analysis would not change significantly for $\mu\neq 0$.
\subsection{Default boundaries}

Following \cite{Lipton2015}, we introduce time-dependent default boundaries $\Lambda_i(t)$. 
Bank $i$ is assumed defaulted if its asset value process crosses its default boundary, such that the default time for bank $i$ is
%Once, we know the default boundaries, we can introduce the default time for each bank
\begin{equation}
	\tau_i = \inf\{t |\, A_i(t) \le \Lambda_i\}, \quad i = 1, 2,
\end{equation}
and we define $\tau = \min(\tau_1, \tau_2)$.

Before any of the banks $i=1,2$ has defaulted, $t<\tau$, 
\begin{equation}
	\Lambda_i = \left\{
	\begin{aligned}
	& R_i (L_i + L_{i \bar{i}}) - L_{\bar{i} i} \equiv \Lambda_i^{<},  \quad & t < T, \\
	&  L_i + L_{i \bar{i}} - L_{\bar{i} i}  \equiv \Lambda_i^{=},  & t = T,
	\end{aligned}
	\right.
\end{equation}
where $0 \le R_i \le 1$ is the recovery rate and $\bar{i} = 3 - i$.

If the $k$-th bank defaults at intermediate time $t$, then for the surviving bank $\bar{k} = 3 -k$ the default boundary changes to
$\Lambda_{\bar{k}}(t+) = \tilde{\Lambda}_{\bar{k}}(t)$, where
\begin{equation}
	\tilde{\Lambda}_{\bar{k}} = \left\{
	\begin{aligned}
	& R_{\bar{k}} (L_{\bar{k}} + L_{\bar{k} k} - R_k L_{k \bar{k}})  \equiv \tilde{\Lambda}_k^{<},  \quad & t < T,\\
	& L_{\bar{k}} + L_{\bar{k} k} - R_k L_{k \bar{k}}   \equiv \tilde{\Lambda}_i^{=},  & t = T.
	\end{aligned}
	\right.
\end{equation}
It is clear that for $\Delta \Lambda_k(t) \equiv \Lambda_k(t+)-\Lambda_k(t)$ we have
\begin{equation}
	\Delta \Lambda_k \equiv  \tilde{\Lambda}_k - \Lambda_k = 
	\begin{cases}
		 (1 - R_{\bar{k}} R_k) L_{k \bar{k}},  &t < T, \\
		 (1 - R_k) L_{k \bar{k}},  &t = T.
	\end{cases}
\end{equation}
Thus, $\Delta  \Lambda_k > 0$ and the corresponding default boundaries move to the right.
This mechanism can therefore trigger cascades of defaults.

\subsection{Terminal conditions}
We need to specify the settlement process at time $t = T$. We shall do this in the spirit of \cite{Eisenberg}. Since at time $T$ full settlement is expected, we assume that bank $i$ will pay the fraction $\omega_i$ of its total liabilities to creditors. This implies that if $\omega_i = 1$ the bank pays all liabilities (both external and interbank) and survives. On the other hand, if $0 < \omega_i < 1$, bank $i$ defaults, and pays only a fraction of its liabilities. Thus, we can describe the terminal condition as a system of equations
\begin{equation}
	\label{term_cond}
	\min \left\{A_i(T) + \omega_{\bar{i}} L_{\bar{i}i}, L_i + L_{i \bar{i}}  \right\} = \omega_i \left(L_i + L_{i \bar{i}} \right).
\end{equation}
There is a unique vector $\omega = (\omega_1, \omega_2)^T$ such that the condition (\ref{term_cond}) is satisfied. See \cite{Lipton2015}, \cite{LiptonItkin2015} for details.

\subsection{Formulation of backward Kolmogorov equation}
For convenience, we introduce normalized dimensionless variables
\begin{equation}
	\bar{t} = \Sigma^2 t, \quad X_i = \frac{\Sigma}{\sigma_i} \ln \left(\frac{A_i}{\Lambda_i^{<}}\right), \quad \bar{\lambda}_i = \frac{\lambda_i}{\Sigma^2},
\end{equation}
where
\begin{equation*}
	\Sigma = \sqrt{\sigma_1 \sigma_2}.
\end{equation*}
Denote also
\begin{equation}
	\xi_i = -\left( \frac{\sigma_i}{2 \Sigma} + \kappa_i \bar{\lambda}_i\right), \quad \zeta_i = \frac{\Sigma}{\sigma_i}.
\end{equation}
Applying Itô's formula to $X_i$, we find its dynamics
\begin{equation}
	d X_i = \xi_i \, d \bar{t} + d W_i(\bar{t}) + \zeta_i J_i \, d N_i(\bar{t}).
\end{equation}
In the following, we omit bars for simplicity.

The default boundaries change to 
\begin{equation}
	\mu_i =
	\begin{cases}
		\mu_i^{<} = 0, & t < T, \\
		\mu_i^{=} = \frac{\Sigma}{\sigma_i} \ln \left(\frac{\Lambda_i^{=}(t)}{\Lambda_i^{<}(t)} \right), & t = T.
	\end{cases}
\end{equation}

Assume that the terminal payoff for a contract is $\psi(X_T)$. Then, the value function is given by
\begin{multline}
	V(x, t) = \mathbb{E} \left[\psi(X_T) \cdot \mathbbm{1}_{\{\tau \ge T \}}  + \int_t^T \chi(s, X_s) \cdot \mathbbm{1}_{\{\tau > s \}}  \, d s \, + \right. \\
	\left. + \phi_{1, 0}(\tau_1, X_2(\tau_1)) \cdot \mathbbm{1}_{\{\tau_1 < T \}} +  \phi_{2, 0}(\tau_2, X_1(\tau_2)) \cdot \mathbbm{1}_{\{\tau_2 < T \}} | \, X(t) = x \right],
\end{multline}
where $\chi(\tau, x)$ is the contract payment at an intermediate time $t \le s \le T$ (for example, coupon payment), and $\phi_{1, 0}(t, X_2(t))$ and  $\phi_{2, 0}(t, X_1(t))$ are the payoffs in case of intermediate default of bank 1 or 2, respectively.

Then, according to the Feynman--Kac formula, the corresponding pricing equation is the Kolmogorov backward equation
\begin{align}
		\label{kolm_1}
		\frac{\partial V}{\partial t} + \mathcal{L} V &= \chi(t, x), \\
		V(t, 0, x_2) &= \phi_{2, 0}(t, x_1), \quad V(t, x_1, x_2) \underset{x_1 \to +\infty}{\longrightarrow} \phi_{2, \infty}(t, x_2), \\
		V(t, x_1, 0) &= \phi_{1, 0}(t, x_2), \quad V(t, x_1, x_2) \underset{x_2 \to +\infty}{\longrightarrow} \phi_{1, \infty}(t, x_1), \\
		\label{kolm_2} V(T, x) &= \psi(x),
\end{align}
where Kolmogorov backward operator
\begin{multline}
	\label{kolmogorov_backward}
	\mathcal{L} f = \frac{1}{2}  f_{x_1 x_1} + \rho f_{x_1 x_2} + \frac{1}{2}  f_{x_2 x_2} +  \xi_1 f_{x_1} + \xi_2 f_{x_2} +  \lambda_{1}  \mathcal{J}_1 f + \lambda_{2}  \mathcal{J}_2 f + \lambda_{12}  \mathcal{J}_{12} f - v  f \\
	= \Delta_{\rho} f + \xi \cdot \nabla f + \mathcal{J} f - v f,
\end{multline}
where $v = \lambda_1 + \lambda_2 + \lambda_{12}$ and
\begin{align}
\mathcal{J}_1 f(x) &= \varsigma_1 \int_0^{x_1} f(x_1 - u,  x_2) e^{-\varsigma_1 u} d u, \label{j1_eq}\\
\mathcal{J}_2 f(x) &= \varsigma_2 \int_0^{x_2} f(x_1,  x_2 - u) e^{-\varsigma_2 u} d u, \label{j2_eq}\\
\mathcal{J}_{12} f(x) &=  \mathcal{J}_1 \mathcal{J}_2 f(x) =  \varsigma_1  \varsigma_2  \int_0^{x_1} \int_0^{x_2} f(x_1 - u, x_2 - v) e^{-\varsigma_1 u - \varsigma_2 v} d u d v \label{j12_eq},
\end{align}
$\varsigma_i = \sigma_i \vartheta_i / \Sigma$, and $\phi_{i, 0}, \phi_{i, \infty}$ are given.

In the following, we formulate the Kolmogorov backward equation for specific quantities.

\subsection{Joint and marginal survival probabilities}
\label{section:joint}
The joint survival probability is the probability that both banks do not default by the terminal time $T$ and given by
\begin{equation}
	Q(t, x) = \mathbb{E}[\mathbbm{1}_{\{\tau \ge T, X_1(T) \ge \mu_1^{=}, X_2(T) \ge \mu_2^{=}\}} \, | X(t) = x].
\end{equation}
Then, applying (\ref{kolm_1})--(\ref{kolm_2}) with $\psi(x) =\mathbbm{1}_{\{x_1 \ge \mu_1^{=}, x_2 \ge \mu_2^{=}\}}$ and $\chi(t, x) = 0$, we get
\begin{equation}
\label{joint_surv_prob}
\begin{aligned}
		& \frac{\partial Q}{\partial t} + \mathcal{L} Q = 0, \\
		& Q(t, x_{1}, 0) = 0, \quad Q(t, 0, x_2) = 0, \\
		& Q(T, x) = \mathbbm{1}_{\{x_1 \ge \mu_1^{=}, x_2 \ge \mu_2^{=}\}}.
\end{aligned}
\end{equation}

The marginal survival probability for the first bank is the probability that the first bank does not default by the terminal time $T$,
\begin{equation}
	q_1(t, x) = \mathbb{E}[\mathbbm{1}_{\{\tau \ge T, X_T \in D_1 \cup D_{12})\}} +  \Xi(\tau_2, X_1(\tau_2)) \cdot \mathbbm{1}_{\{\tau_2 < T \}}| \, X(t) = x],
\end{equation}
where $D_{12}$ is the set where both banks survive at the terminal time, $D_1$ is the set where only the first bank survives, and $ \Xi(\tau_2, X_1(\tau_2)) $ is the one-dimensional survival probability with the modified boundaries.

Then, applying (\ref{kolm_1})--(\ref{kolm_2}) with $\psi(x) = \mathbbm{1}_{\{x \in D_1 \cup D_{12})\}}, \chi(t, x) = 0$, we get
\begin{equation}
\begin{aligned}
		& \frac{\partial}{\partial t} q_1(t, x) + \mathcal{L} q_1(t, x) = 0, \\
		& q_1(t, 0, x_2) = 0, \quad
		q_1(t, x_1, 0) = \Xi(t, x_1) = 
		\begin{cases}
			\chi_{1,0}(t, x_1), & x_1 \ge \tilde{\mu}_1^{<}, \\
			0, & x_1 < \tilde{\mu}_1^{<},
		\end{cases} \\
		& q_1(t, \infty, x_2) = 1, \quad
		q_1(t, x_1, \infty) = 
			\chi_{1,\infty}(t, x_1), \\
		& q_1(T, x) = \mathbbm{1}_{\{x \in D_1 \cup D_{12} \}}.
\end{aligned}
\end{equation}
The function $\chi_{1, 0}(t, x_1)$ is the 1D survival probability, which solves the following boundary value problem
\begin{equation}
	\begin{aligned}
		& \frac{\partial}{\partial t} \chi_{1, 0}(t, x_1) + \mathcal{L}_1 \chi_{1, 0}(t, x_1)= 0, \\
		& \chi_{1, 0}(t, \tilde{\mu}_1^{<}) = 0, \quad \chi_{1, 0}(t, \infty) = 1, \\
		& \chi_{1, 0}(T, x_1) = \mathbbm{1}_{\{x_1 > \tilde{\mu}_1^{=}\}},
	\end{aligned}
\end{equation}
where
\begin{equation*}
	\mathcal{L}_1 f = \frac{1}{2} \frac{\partial^2}{\partial x_1^2} f + \xi_1 \frac{\partial}{\partial x_1} f + \lambda_1 \mathcal{J}_1 f - \lambda_1 f.
\end{equation*}

Accordingly, $\chi_{1, \infty}(t, x_1)$ is the 1D survival probability that solves the following boundary value problem
\begin{equation}
	\begin{aligned}
		& \frac{\partial}{\partial t} \chi_{1, \infty}(t, x_1) + \mathcal{L}_1 \chi_{1, \infty}(t, x_1)= 0, \\
		& \chi_{1, \infty}(t,  0) = 0, \quad \chi_{1, \infty}(t, \infty) = 1, \\
		& \chi_{1, \infty}(T, x_1) = \mathbbm{1}_{\{x_1 > \mu_1^{=}\}}.
	\end{aligned}
\end{equation}
We formulate the pricing problems for CDS, FTD, CVA and DVA in Appendix A.

\section{Numerical scheme}

We shall solve the PIDE \eqref{kolm_1}--\eqref{kolm_2} numerically with an Alternating Direction Implicit (ADI) method. The scheme is a modification of \cite{LiptonSepp} that is unconditionally stable and has second order of convergence in both time and space step.

In order to deal with a forward equation instead of a backward equation, we change the time variable to $\tau = T - t$, so that
\begin{equation}
	\label{pide_forward}
	\begin{aligned}
		& \frac{\partial V}{\partial \tau} = \mathcal{L} V(\tau, x_1, x_2) - \chi(\tau, x_1, x_2), \\
		& V(\tau, x_1, 0) = \phi_{0, 1}(\tau, x_1), \quad V(\tau, 0, x_2) =  \phi_{0, 2}(\tau, x_2), \\
		& V(\tau, x_1, x_2)  \underset{x_2 \to +\infty}{\longrightarrow} \phi_{\infty, 1}(\tau, x_1), \quad V(\tau, x_2, x_2)  \underset{x_1 \to +\infty}{\longrightarrow}  \phi_{\infty, 2}(\tau, x_2), \\
		& V(0, x_1, x_2) = \psi(x_1, x_2).
	\end{aligned}
\end{equation}

We consider the same grid for integral and differential part of the equation
\begin{equation}
	\begin{aligned}
		0 = x_1^0 < x_1^1 < \ldots < x_1^{m_1}, \\
		0 = x_2^0 < x_2^1 < \ldots < x_2^{m_2},
	\end{aligned}
\end{equation}
where $x_1^{m_1}$ and $x_2^{m_2}$ are large positive numbers.

The grid is non-uniform, and is chosen such that relatively many points lie near the default boundaries for better precision. We use a method similar to \cite{itkin2011jumps} to construct the grid.
\subsection{Discretization of the integral part of the PIDE}
In this section, we shall show how to deal with the integral part of the PIDE, and develop an iterative algorithm for the fast computation of the integral operator on the grid. To this end, we outline the scheme from \cite{LiptonSepp} and then give a new method.

The first approach is to deal with the integral operators directly. After the approximation of the integral, we get (\cite{LiptonSepp})
\begin{align}
	&\mathcal{J}_1 V(x_1 + h, x_2) = e^{-\varsigma_1 h} \mathcal{J}_1 V(x_1, x_2) +  \omega_0(\varsigma_1, h) V(x_1, x_2) + \omega_1(\varsigma_1, h) V(x_1 + h, x_2) + O(h^3), \label{J_1_approx}\\
	& \mathcal{J}_2 V(x_1, x_2 + h) = e^{-\varsigma_2 h} \mathcal{J}_2 V(x_1, x_2) +  \omega_0(\varsigma_2, h) V(x_1, x_2) + \omega_1(\varsigma_2, h) V(x_1, x_2 + h) + O(h^3) \label{J_2_approx},
\end{align}
where
\begin{equation*}
	\omega_0(\varsigma, h) = \frac{1 - (1 + \varsigma h) e^{-\varsigma h}}{\varsigma h}, \quad \omega_1(\varsigma, h) = \frac{-1 + \varsigma h + e^{-\varsigma h}}{\varsigma h}.
\end{equation*}

We can also approximate $\mathcal{J}_{12} V = \mathcal{J}_1 \mathcal{J}_2 V$ by applying above approximations for $\mathcal{J}_1$ and $\mathcal{J}_2$ consecutively.

Consider the grid
\begin{equation}
	\begin{aligned}
		0 = x_1^0 < x_1^1 < \ldots < x_1^{m_1}, \\
		0 = x_2^0 < x_2^1 < \ldots < x_2^{m_2},
	\end{aligned}
\end{equation}
where $x_1^{m_1}$ and $x_2^{m_2}$ are large positive numbers.\\

Then, we can write recurrence formulas for computing the integral operator on the grid. Denote $J_1^{i, j}, J_2^{i, j}$, $J_{12}^{i, j}$ the corresponding approximations of $\mathcal{J}_{1}V(x_1^i, x_2^j)$ , $ \mathcal{J}_{2}V(x_1^i, x_2^j)$, $\mathcal{J}_{12}V(x_1^i, x_2^j)$ on the grid. Applying (\ref{J_1_approx}) and (\ref{J_2_approx}) we get

\begin{align}
	&J_1^{i+1, j} = e^{-\varsigma_1 h^1_{i+1}} J_1^{i, j} + \omega_0(\varsigma_1, h^1_{i+1}) V(x_1^i, x_2^j) + \omega_1(\varsigma_1, h_{i+1}^1) V(x_1^{i+1}, x_2^{j}), \label{J_1_rec}\\
	&J_2^{i, j+1} = e^{-\varsigma_2 h^2_{j+1}} J_2^{i, j} + \omega_0(\varsigma_2, h^2_{j+1}) V(x_1^i, x_2^j) + \omega_1(\varsigma_2, h_{j+1}^2) V(x_1^{i}, x_2^{j+1}), \label{J_2_rec}
\end{align}
where $h_{i+1}^1 = x_1^{i+1} - x_1^{i}, h_{j+1}^2 = x_2^{j+1} - x_2^{j}$.

For an alternative method, we rewrite the integral operator as a differential equation
\begin{align}
		& \frac{\partial}{\partial x_1} \left(\mathcal{J}_1 V(x_1, x_2)e^{\varsigma_1 x_1} \right) =  \varsigma_1 V(x_1, x_2) e^{\varsigma_1 x_1}, \label{J1_ode}\\
		& \frac{\partial}{\partial x_2} \left(\mathcal{J}_2 V(x_1, x_2)e^{\varsigma_2 x_2} \right) =  \varsigma_2 V(x_1, x_2) e^{\varsigma_2 x_2}, \label{J2_ode}\\
		& \frac{\partial^2}{\partial x_1 \partial x_2} \left(\mathcal{J}_{12} V(x_1, x_2)e^{\varsigma_1 x_1 + \varsigma_2 x_2} \right) =  \varsigma_1  \varsigma_2 V(x_1, x_2) e^{\varsigma_1 x_1 + \varsigma_2 x_2}. \label{J12_pde}
\end{align}
Then, we apply the Adams-Moulton method of second order which gives us third order of accuracy locally (\cite{butcher2008numerical})
\begin{align}
	&J_1^{i+1, j} = e^{-\varsigma_1 h^1_{i+1}} J_1^{i, j} + \frac{1}{2} h^1_{i+1} e^{-\varsigma_1 h^1_{i+1}} \varsigma_1  V(x_1^i, x_2^j) + \frac{1}{2} h^1_{i+1} \varsigma_1 V(x_1^{i+1}, x_2^{j}), \label{J_1_rec_adams}\\
	&J_2^{i, j+1} = e^{-\varsigma_2 h^2_{j+1}} J_2^{i, j} +\frac{1}{2} h^2_{j+1} e^{-\varsigma_2 h^2_{j+1}} \varsigma_2  V(x_1^i, x_2^j) + \frac{1}{2} h^2_{j+1} \varsigma_2 V(x_1^{i}, x_2^{j+1}), \label{J_2_rec_adams}
\end{align}
where $h_{i+1}^1 = x_1^{i+1} - x_1^{i}, h_{j+1}^2 = x_2^{j+1} - x_2^{j}$, and is equivalent to the trapezoidal rule.

We can rewrite (\ref{J_1_rec_adams})--(\ref{J_2_rec_adams}) in the same notation as  (\ref{J_1_rec})--(\ref{J_2_rec}) by defining
\begin{equation*}
	\omega_0(\varsigma, h) = \frac{1}{2} h e^{-\varsigma h} \varsigma, \quad \omega_1(\varsigma, h) = \frac{1}{2} h \varsigma.
\end{equation*}
So,
\begin{align*}
	&J_1^{i+1, j} = e^{-\varsigma_1 h^1_{i+1}} J_1^{i, j} + \omega_0(\varsigma_1, h^1_{i+1}) V(x_1^i, x_2^j) + \omega_1(\varsigma_1, h_{i+1}^1) V(x_1^{i+1}, x_2^{j}),\\
	&J_2^{i, j+1} = e^{-\varsigma_2 h^2_{j+1}} J_2^{i, j} + \omega_0(\varsigma_2, h^2_{j+1}) V(x_1^i, x_2^j) + \omega_1(\varsigma_2, h_{j+1}^2) V(x_1^{i}, x_2^{j+1}). 
\end{align*}

As a result we get explicit recursive formulas for approximations of $\mathcal{J}_1 V$ and $\mathcal{J}_2 V$ that can be computed for all grid points via $O(m_1 m_2)$ operations. Both methods give the same order of accuracy. As was discussed above, in order to compute the approximation of $\mathcal{J}_{12} V$ we can apply consecutively the approximations of $\mathcal{J}_2 V$ and $\mathcal{J}_1 (\mathcal{J}_2 V)$. So, we have the two-step procedure:
\begin{equation}
	I_{12}^{i+1, j} = e^{-\varsigma_1 h^1_{i+1}} I_{12}^{i, j} + \omega_0(\varsigma_1, h^1_{i+1}) V(x_1^i, x_2^j) + \omega_1(\varsigma_1, h_{i+1}^1) V(x_1^{i+1}, x_2^{j}), \label{I_12_rec}\\
\end{equation}
and 
\begin{equation}
	J_{12}^{i, j+1} = e^{-\varsigma_2 h^2_{j+1}} J_{12}^{i, j} + \omega_0(\varsigma_2, h^2_{j+1}) I_{12}^{i, j} + \omega_1(\varsigma_2, h_{j+1}^2) I_{12}^{i, j+1}. \label{J_12_rec}\\
\end{equation}
Using this two-step procedure, we can also compute an approximation of $\mathcal{J}_{12} V$ on the grid in complexity $O(m_1 m_2)$.

 We shall subsequently analyze the stability of the second method and use it in the numerical tests. The results for the first method would be very similar.
 
% \paragraph{Eigenvalues of discretized jump operator.}
%\label{jump_eigs}
For the implementation, computing and storing a matrix representation of the jump operator is not necessary, since the operator can be computed iteratively as described above, but we shall use matrix notation for the analysis. We henceforth denote $J_1, J_2$, and $J_{12}$ the matrices of the discretized jump operators. From (\ref{J_1_rec})--({\ref{J_2_rec}) we can find that the matrices $J_1$ and $J_2$ are lower-triangular with diagonal elements $w_1 = \omega_1(\varsigma_1, h_1)$ and $w_2 =  \omega_1(\varsigma_2, h_2)$. Then, $J_{12} = J_1 J_2$ is also a lower-triangular matrix with diagonal elements $w_1 w_2$. To illustrate, in Figure \ref{jump_matrices} we plot the sparsity patterns in $J_1, J_2$, and $J_{12}$.
 \begin{figure}[H]
	\begin{center}
				\subfloat[%Sparsity pattern of 
				$J_1$.]{\includegraphics[width=0.32\textwidth,trim={2cm 0 2cm 0.5cm},clip]{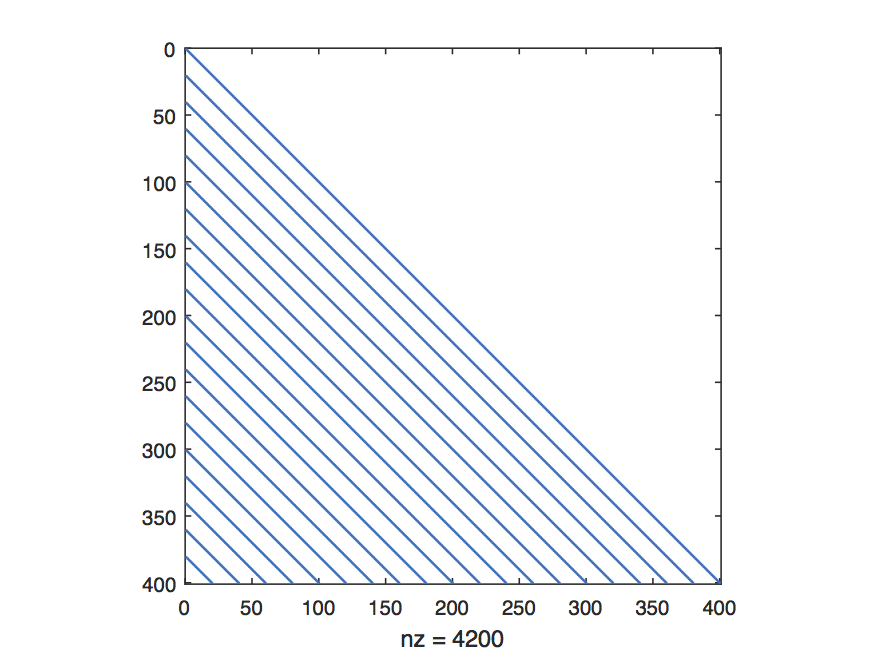}}
				\subfloat[%Sparsity pattern of 
				$J_2$.]{\includegraphics[width=0.32\textwidth,trim={2cm 0 2cm 0.5cm},clip]{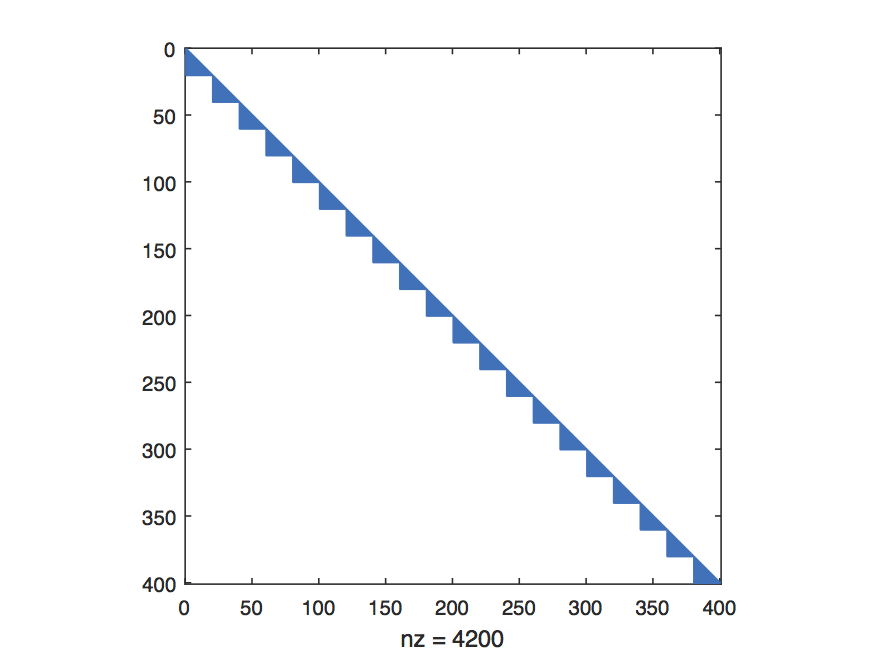}}
				\subfloat[%Sparsity pattern of 
				$J_{12}$.]{\includegraphics[width=0.32\textwidth,trim={2cm 0 2cm 0.5cm},clip]{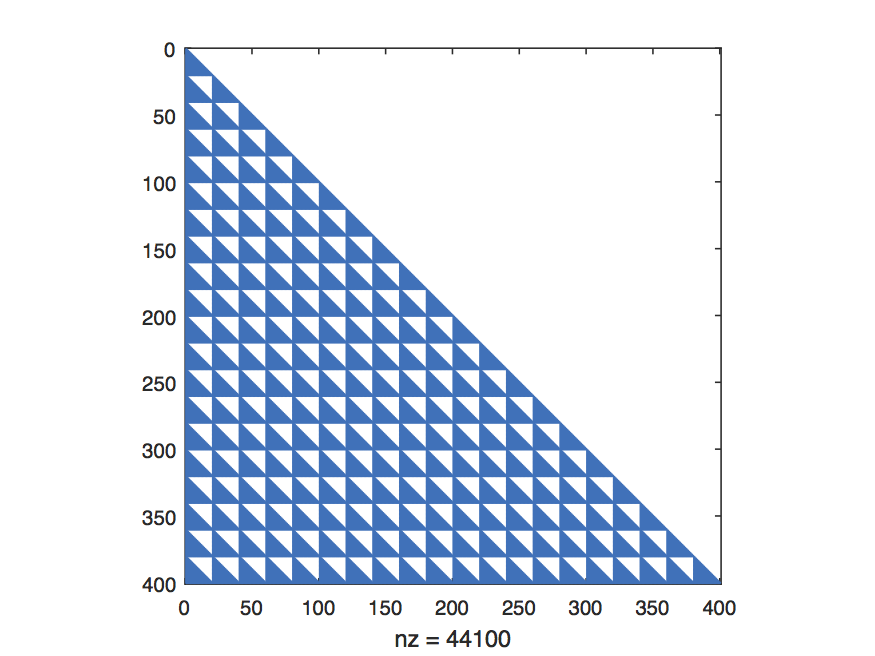}}
	\end{center}		
	\vspace{-20pt}
	\caption{Sparsity pattern of $J_1$, $J_2$, and $J_{12}$. Here, $m_1=m_2=20$ and $nz$ is the number of non-zero elements of the matrices.}
 	\label{jump_matrices}
\end{figure}
%Since the matrices are lower-triangular, the eigenvalues are the diagonal elements. Thus, matrix $J_1$ has all eigenvalues equal to $w_1$, matrix $J_2$ has all eigenvalues equal to $w_2$, and matrix $J_{12}$ has all eigenvalues equal to $w_{12} = w_1 w_2$.

\subsection{Discretization of the differential part of the PIDE}
Now consider the approximation of derivatives in the differential operator on a non-uniform grid. We use the standard derivative approximation (\cite{kluge2002}, \cite{in2010adi}). For the first derivative over each variable consider right-sided, central, and left-sided schemes. So, for the derivative over $x_1$ we have:
\begin{align}
&	\frac{\partial V}{\partial x_1}(x_1^i, x_2^j) \approx \alpha^1_{i, -2} V(x_1^{i-2}, x_2^j) + \alpha^1_{i, -1} V(x_1^{i-1}, x_2^j)+ \alpha^1_{i, 0} V(x_1^i, x_2^j), \label{D_x1_1}\\
&	\frac{\partial V}{\partial x_1}(x_1^i, x_2^j) \approx \beta^1_{i, -1} V(x_1^{i-1}, x_2^j) + \beta^1_{i, 0} V(x_1^{i}, x_2^j)+ \beta^1_{i, 1} V(x_1^{i+1}, x_2^j), \label{D_x1_center}\\
&	\frac{\partial V}{\partial x_1}(x_1^i, x_2^j) \approx \gamma^1_{i, 0} V(x_1^{i}, x_2^j) + \gamma^1_{i, 1} V(x_1^{i+1}, x_2^j)+ \gamma^1_{i, 2} V(x_1^{i+2}, x_2^j) \label{D_x1_2},
\end{align}
while for derivative over $x_2$ we have:
\begin{align}
&	\frac{\partial V}{\partial x_2}(x_1^i, x_2^j) \approx \alpha^2_{j, -2} V(x_1^i, x_2^{j-2}) + \alpha^2_{j, -1} V(x_1^i, x_2^{j-1})+ \alpha^2_{j, 0} V(x_1^i, x_2^j), \label{D_x2_1} \\
&	\frac{\partial V}{\partial x_2}(x_1^i, x_2^j) \approx \beta^2_{j, -1} V(x_1^i, x_2^{j-1}) + \beta^2_{j, 0} V(x_1^{i}, x_2^j)+ \beta^2_{j, 1} V(x_1^j, x_2^{j+1}), \label{D_x2_center}\\
&	\frac{\partial V}{\partial x_2}(x_1^i, x_2^j) \approx \gamma^2_{j, 0} V(x_1^{i}, x_2^j) + \gamma^2_{j, 1} V(x_1^i, x_2^{j+1}, x_2^j)+ \gamma^2_{j, 2} V(x_1^i, x_2^{j+2}), \label{D_x2_2}
\end{align}
with coefficients
\begin{equation*}
	\begin{aligned}
		\alpha^k_{i, -2} &= \frac{\Delta x_k^i}{\Delta x_k^{i-1} (\Delta x_k^{i-1} + \Delta x_k^i)},  & \alpha^k_{i, -1} &= \frac{-\Delta x_k^{i-1} - \Delta x_k^i}{\Delta x_k^{i-1} \Delta x_k^i},  &\alpha^k_{i, 0} &= \frac{\Delta x_k^{i-1} + 2 \Delta x_k^i}{\Delta x_k^i (\Delta x_k^{i-1} + \Delta x_k^i)} , \\
		\beta^k_{i, -1} &= \frac{-\Delta x_k^{i+1}}{\Delta x_k^{i} (\Delta x_k^{i} + \Delta x_k^{i+1})}, & \beta^k_{i, 0} &= \frac{\Delta x_k^{i+1} - \Delta x_k^i}{\Delta x_k^{i} \Delta x_k^{i+1}},  &\beta^k_{i, 1} &= \frac{\Delta x_k^{i}}{\Delta x_k^{i+1} (\Delta x_k^{i} + \Delta x_k^{i+1})} , \\
		\gamma^k_{i, 0} &= \frac{-2\Delta x_k^{i+1} - \Delta x_k^{i+2}}{\Delta x_k^{i+1} (\Delta x_k^{i+1} + \Delta x_k^{i+2})}, & \gamma^k_{i, 1} &= \frac{\Delta x_k^{i+1} + \Delta x_k^{i+2}}{\Delta x_k^{i+1} \Delta x_k^{i+2}},  & \gamma^k_{i, 2} &= \frac{-\Delta x_k^{i+1}}{\Delta x_k^{i+2} (\Delta x_k^{i+1} + \Delta x_k^{i+2})} .
	\end{aligned}
\end{equation*}
For the boundaries at $0$ we use the schemes \eqref{D_x1_1} and \eqref{D_x2_1}, for the right boundaries at $x_1^{m_1}$ and $x_2^{m_2}$ we use the schemes \eqref{D_x1_2} and \eqref{D_x2_2}, and for other points we use the central schemes \eqref{D_x1_center} and \eqref{D_x2_center}.

To approximate the second derivative we use the central scheme:
\begin{align}
	&	\frac{\partial^2 V}{\partial x_1^2}(x_1^i, x_2^j) \approx \delta^1_{i, -1} V(x_1^{i-1}, x_2^j) + \delta^1_{i, 0} V(x_1^{i}, x_2^j)+ \delta^1_{i, 1} V(x_1^{i+1}, x_2^j), \label{D2_x1} \\
&	\frac{\partial^2 V}{\partial x_2^2}(x_1^i, x_2^j) \approx \delta^2_{j, -1} V(x_1^i, x_2^{j-1}) + \delta^2_{j, 0} V(x_1^{i}, x_2^j)+ \delta^2_{j, 1} V(x_1^j, x_2^{j+1}) \label{D2_x2},
\end{align}
with coefficients
\begin{equation*}
		\delta^k_{i, -1} = \frac{2}{\Delta x_k^{i} (\Delta x_k^{i} + \Delta x_k^{i+1})}, \quad \delta^k_{i, 0} = \frac{-2}{\Delta x_k^{i} \Delta x_k^{i+1}}, \quad \delta^k_{i, 1} = \frac{2}{\Delta x_k^{i+1} (\Delta x_k^{i} + \Delta x_k^{i+1})},
\end{equation*}	
and to approximate the second mixed derivative we use the scheme:
\begin{equation}
	\frac{\partial^2 V}{\partial x_1 \partial x_2} (x_1^i, x_2^j) \approx \sum_{k, l = -1}^1 \beta_{i, k}^1 \beta_{j, l}^2 V(x_1^{i+k}, x_2^{j+l}). \label{D_x1x2}	
\end{equation}

As a result, we can approximate the differential operator $\mathcal{D} V$ by a discrete operator
\begin{equation}
	D V = D_1 V + D_2 V + D_{12} V,
\end{equation}
where $D_1 V$ contains the discretized derivatives over $x_1$ defined in (\ref{D_x1_1})--(\ref{D_x1_2}) and (\ref{D2_x1}), $D_2 V$ contains the discretized derivatives over $x_2$ defined in (\ref{D_x2_1})--(\ref{D_x2_2}) and (\ref{D2_x2}), and $D_{12} V$ contains the discretized mixed derivative defined in (\ref{D_x1x2}).

By straightforward but lengthy Taylor expansion of the expression in (\ref{D_x1_1})--(\ref{D_x1x2}), the scheme (\ref{HV_scheme}) has  second order truncation error in variables $x_1$ and $x_2$ for meshes which are either uniform or smooth transformations of such meshes, as we shall consider later.

\subsection{Time discretization: ADI scheme}
After discretization over $(x_1, x_2)$ we can rewrite PIDE (\ref{pide_forward}) as a system of ordinary (linear) differential equations. Consider the vector $U(t) \in \mathbb{R}^{m_1m_2 \times 1}$ whose elements correspond to $V(t, x_1^i, x_2^j)$. Then
\begin{equation}
	\begin{aligned}
		& U'(t) = \tilde{A} U(t) + b(t), \\
		& U(0) = U_0,
	\end{aligned}
\end{equation}
where $\tilde{A} = D_1 + D_2 + D_{12} + \lambda_1 J_1 + \lambda_2 J_2 + \lambda_{12} J_{12} - (\lambda_1 + \lambda_2 + \lambda_{12}) I$, and $b(t)$ is determined from boundary conditions and the right-hand side.

To solve this system, we apply an ADI scheme for the time discretization. Consider, for simplicity, a uniform time mesh with time step $\Delta t: t_n = n \Delta t, n = 0, \ldots, N-1$. 

We decompose the matrix $\tilde{A}$  into three matrices, $\tilde{A} = \tilde{A}_0 + \tilde{A}_1 + \tilde{A}_2$, where
\begin{align*}
	& \tilde{A}_0 =  D_{12} + \lambda_1 J_1 + \lambda_2 J_2 + \lambda_{12} J_{12},  \\
  	& \tilde{A}_1 = D_1 - \left(\lambda_1 + \frac{\lambda_{12}}{2} \right) I, \\
	  & \tilde{A}_2 = D_2 - \left(\lambda_2 + \frac{\lambda_{12}}{2} \right) I,
\end{align*}
and $b(t) = b_0(t) + b_1(t) + b_2(t)$, where $b_0(t)$ corresponds to the right-hand side and the FD discretization of the mixed derivatives on the boundary, $b_1(t)$ and $b_2(t)$ correspond to the FD discretization of the derivatives over $x_1$ and $x_2$ on the boundary.

Now we can apply a traditional ADI scheme with matrices $\tilde{A}_0, \tilde{A}_1$, and $\tilde{A}_2$. We choose the Hundsdorfer--Verwer (HV) scheme (\cite{HV}) in order to have second order accuracy in the time variable, and unconditional stability, as we shall prove below. For convenience, denote
\begin{align}
	& F_j(t, x) = \tilde{A}_j x + b_j(t), \quad j = 0, 1, 2, \label{F_j}\\
	& F(t, x) = (\tilde{A}_0 + \tilde{A}_1 + \tilde{A}_2 ) x + (b_0(t) + b_1(t) + b_2(t)),
\end{align}
and apply the Hundsdorfer--Verwer (HV) scheme:
\begin{equation}
	\label{HV_scheme}
	\left\{
	\begin{aligned}
	&	Y_0 = U_{n-1} + \Delta t F(t_{n-1}, U_{n-1}), \\
	&	Y_j = Y_{j-1} + \theta \Delta t (F_j(t_n, Y_j) - F_j(t_n, U_{n-1})), \quad j = 1, 2, \\
	&	\tilde{Y}_0 = Y_0 + \sigma \Delta t (F(t_n, Y_2) - F(t_{n-1}, U_{n-1})), \\
	&	\tilde{Y}_j = \tilde{Y}_{j-1} + \theta \Delta t (F_j(t_n, \tilde{Y}_j - F_j(t_n, Y_2)), \quad j = 1, 2, \\
	&	U_n = \tilde{Y}_2.
	\end{aligned}
	\right.
\end{equation}

In this scheme, parts that contain $F_1$ and $F_2$ are treated implicitly. The matrix $\tilde{A}_1$ is tridiagonal and $\tilde{A}_2$ is block-tridiagonal and can be inverted via $O(m_1 m_2)$ operations.  As a result, the overall complexity is $O(m_1 m_2)$ for a single time step or $O(N m_1 m_2)$ for the whole procedure.

Moreover, the scheme has second order of consistency in both $(x_1, x_2)$ and $t$ for any given $\theta$ and $\sigma = \frac{1}{2}$. 

\subsection{Stability analysis}
In this section, we consider the PIDE \eqref{pide_forward} with zero boundary conditions at $0$ in both directions and on a uniform grid,
such that $F_j(t, x) = \tilde{A}_j x$ and
\begin{equation}
	\label{HV_nobound}
	\left\{
	\begin{aligned}
	&	Y_0 = U_{n-1} + \Delta t \tilde{A} U_{n-1}, \\
	&	Y_j = Y_{j-1} + \theta \Delta t (\tilde{A}_j Y_j - \tilde{A}_j U_{n-1}),\quad j = 1, 2, \\
	&	\tilde{Y}_0 = Y_0 + \sigma \Delta t (\tilde{A} Y_2- \tilde{A} U_{n-1}), \\
	&	\tilde{Y}_j = \tilde{Y}_{j-1} + \theta \Delta t (\tilde{A}_j\tilde{Y}_j - \tilde{A}_j Y_2), \quad j = 1, 2 \\
	&	U_n = \tilde{Y}_2.
	\end{aligned}
	\right.
\end{equation}

For convenience, we denote by $F: U_n = F U_{n-1}$.

We further consider the PDE on $\mathbb{R}^2$, i.e., without default boundaries. Hence, we assume that diffusion and jump operators are discretized on 
 an infinite, uniform mesh $\{(j_1 h_1, j_2 h_2), (j_1, j_2) \in \mathbb{Z}^2\}$, such that, e.g.\ $D_1, D_2, D_{12}, J_1, J_2$ are infinite matrices.
 This is different to \cite{intHoutStability}, where finite matrices and periodic boundary conditions (without integral terms) are considered.

We use von Neumann stability analysis, as first introduced by \cite{charney1950numerical}, by expanding the solution into a Fourier series.
Hence, we shall show that the proposed scheme (\ref{HV_nobound}) is unconditionally stable,
 i.e.\ we will show that all eigenvalues of the operator  $F$ have moduli bounded by 1 plus an $O(\Delta t)$ term,
 where the corresponding eigenfunctions are given by $\exp(i \phi_1 j_1) \exp(i \phi_2 j_2)$, with $\phi_1$ and $\phi_2$ the wave numbers and
 $j_1$ and $j_2$ the grid coordinates.
 
%\paragraph{Stability analysis of scheme (\ref{HV_nobound})}
 \cite{intHoutStability} show that when all matrices commute
(as in the PDE case with periodic boundary conditions), 
the eigenvalues for $F$ are given by %this leads to the condition %(\cite{intHoutStability})
%\begin{equation}
%	\label{stab_eq}
%	|T(\tilde{z}_0, \tilde{z}_1, \tilde{z}_2)| \le 1,
%\end{equation}
\begin{eqnarray}
\label{defT}
T(\tilde{z}_0, \tilde{z}_1, \tilde{z}_2) &=& 1 + 2 \frac{\tilde{z}_0 + \tilde{z}}{p} - \frac{\tilde{z}_0 + \tilde{z}}{p^2} + \sigma \frac{(\tilde{z}_0 + \tilde{z})^2}{p^2} \quad \text{with} \\
	p &=& (1 - \theta \tilde{z}_1) (1 - \theta \tilde{z}_2), \nonumber
\end{eqnarray}
where $\tilde{z}_j = \tilde{\mu}_j \Delta t$, where $\tilde{\mu}_j$ is an eigenvalue of $\tilde{A}_j$, $j = 0, 1, 2$, $\tilde{z} = \tilde{z}_1 + \tilde{z}_2$, $\theta \ge 0$.

The analysis is made slightly more complicated in our case through the presence of the  jump operators.
In the remainder of this section, %we first show that for our case \eqref{stab_eq} can still be applied, and then 
we show that stability is still given under the same conditions on $\theta$ and $\sigma$ as in the purely diffusive case. For the correspondence of notation with \cite{intHoutStability}, we denote $A = A_0 + A_1 + A_2$, where $A_0 = D_{12}, A_1 = D_1$, $A_2 = D_2$ and $\mu_0, \mu_1$, and $\mu_2$ are the eigenvalues of the corresponding matrices. Similar to $\tilde{z}_0, \tilde{z}_1,$ and $\tilde{z}_2$, we define scaled eigenvalues $z_0 = \mu_0 \Delta t, z_1 = \mu_1 \Delta t, z_2 = \mu_2 \Delta t$.

%Using properties of lower-triangular matrices, we can easily see that the matrix $F$ from \eqref{HV_nobound} can be represented as $F = U T_F U^{*}$, where $T_F$ is a lower-triangular matrix, whose eigenvalues are equal to \eqref{defT}.

%\begin{lemma}
%\label{lemma_commute}
%The following identities are satisfied
We have the eigenvalues $\tilde{\mu}_j$ of $\tilde{A}_j$ given by
\begin{align}
	& \tilde{\mu}_0 = \mu_0 + \lambda_1 w_1 + \lambda_2 w_2 + \lambda_{12} w_{12}, \label{mu_0_eq} \\
	& \tilde{\mu}_1 = \mu_1 - \left(\lambda_1 + \frac{\lambda_{12}}{2}\right), \label{mu_1_eq}\\
	& \tilde{\mu}_2 = \mu_2 - \left(\lambda_2 + \frac{\lambda_{12}}{2}\right) \label{mu_2_eq},
\end{align}
where $\mu_j$ is  an eigenvalue of $A_j$, and $w_1, w_2$, and $w_{12}$ are eigenvalues of $J_1, J_2$, and $J_{12}$.
%\end{lemma}

Denote % $z_j = \mu_j \Delta t$ with $\mu_j$  an eigenvalue of $A_j$, 
$z = z_1 + z_2$, $s_1 = w_1  \Delta t, s_2 = w_2 \Delta t, s_{12} = w_{12} \Delta t$, where $w_1, w_2, w_{12}$ are eigenvalues of $J_1, J_2, J_{12}$ respectively, and $s_0 = \lambda_1 s_1 + \lambda_2 s_2 + \lambda_{12} s_{12}$.

	Multiplying (\ref{mu_0_eq})--(\ref{mu_2_eq}) by $\Delta t$, we have 
	\begin{align} 
		& \tilde{z}_0 = z_0 + s_0, \label{tilde_z0} \\
		& \tilde{z}_1 = z_1 - \left(\lambda_1 + \frac{\lambda_{12}}{2}\right) \Delta t, \\
		& \tilde{z}_2 = z_2 - \left(\lambda_2 + \frac{\lambda_{12}}{2}\right) \Delta t.  \label{tilde_z2} 		
	\end{align}
\begin{theorem}[\cite{intHoutStability}, Theorem 3.2]
	\label{theor_inthout}
	Assume $\Re({z}_1) \le 0, \Re({z}_2) \le 0$, $|{z}_0| \le 2\sqrt{\Re({z}_1) \Re({z}_2)}$, where ${z}_0, {z}_1$, and ${z}_2$ are the eigenvalues of ${A}_0, {A}_1$, and ${A}_2$, %in \eqref{F_j} 
	and
	\begin{equation*}
		\frac{1}{2} \le \sigma \le \left(1 + \frac{\sqrt{2}}{2} \right) \theta.
	\end{equation*}
	Then,
	\begin{equation*}
		|T({z}_0, {z}_1, {z}_2)| \le 1,
	\end{equation*}
	and the Hundsdorfer--Verwer scheme \eqref{HV_nobound} is stable in the purely diffusive case.
\end{theorem}

\begin{lemma}
The scaled eigenvalues of $A_0$, $A_1$, $A_2$, $J_1$, $J_2$, $J_{12}$ can be expressed as
\begin{eqnarray}
\label{z0}
	z_0 &=& -\rho b [\sin{\phi_1} \sin{\phi_2}],  \\
	\label{z1}
	z_1 &=& -a_1 (1 - \cos{\phi_1}) + i \xi_1 q_1 \sin{\phi_1}, \\
	\label{z2}
	z_2 &=& -a_2 (1 - \cos{\phi_2}) + i \xi_2 q_2 \sin{\phi_2}, \\
%	s_1 &=& \Delta t \, \zeta_1 h_1 \left(\frac{1}{2} + \frac{\exp(-h_1 (\zeta_1 + i \phi_1))}{1-\exp(-h_1 (\zeta_1 + i \phi_1))} \right), \\
%	s_2 &=& \Delta t \, \zeta_2 h_2 \left(\frac{1}{2} + \frac{\exp(-h_2 (\zeta_2 + i \phi_2))}{1-\exp(-h_2 (\zeta_2 + i \phi_2))} \right), \\
	s_1 &=& \Delta t \, \zeta_1 h_1 \left(\frac{1}{2} + \frac{\exp(-h_1 \zeta_1 + i \phi_1)}{1-\exp(-h_1 \zeta_1 + i \phi_1)} \right), \\
	s_2 &=& \Delta t \, \zeta_2 h_2 \left(\frac{1}{2} + \frac{\exp(-h_2 \zeta_2 + i \phi_2)}{1-\exp(-h_2 \zeta_2 + i \phi_2)} \right), \\
	s_{12} &=& s_1 s_2/\Delta t,
\end{eqnarray}
where
\begin{equation*}
	q_1 = \frac{\Delta t}{h_1}, \quad q_2 = \frac{\Delta t}{h_2}, \quad a_1 = \frac{\Delta t}{h_1^2}, \quad a_2 = \frac{\Delta t}{h_2^2}, \quad b= \frac{\Delta t}{h_1 h_2},
\end{equation*}
and $\phi_j \in [0, 2 \pi]$ for $j = 1, 2$.

Moreover,
\begin{equation}
\label{karelineq}
	|z_0| \le 2 \sqrt{\Re(z_1) \Re(z_2)}.
\end{equation}	
\end{lemma}
\begin{proof}
	All six eigenvalues follow by insertion of the ansatz $U=\exp(i \phi_1 j_1) \exp(i \phi_2 j_2)$.
	For instance, %for $U(j,k)=$
	\[
	(J_1 U)(j_1,j_2) = \zeta_1 h_1 \left(\frac{1}{2} U(j_1,j_2) +  \sum_{k=1}^\infty \exp(-\zeta_1 h_1 k) U(j_1-k,j_2) \right),
	\]
	and the result follows by using the special form of $U$ and evaluating the geometric series.
	
	Alternatively, the first three equations follow immediately from the eigenvalues for finite matrices (\cite{intHoutStability}, p.29),
	which are given by (\ref{z0})--(\ref{z2}) where $\phi_j = 2 l \pi/m_j$, $l=1,\ldots,m_j$.
	In the infinite mesh case, the spectrum is the continuous limit and (\ref{karelineq}) still holds.
%	 then the eigenvalues of the semi-infinite (banded Toeplitz) matrix are given by the Schmidt and Spitzer Theorem (see Theorem 11.17 in \cite{bottcher2005spectral}) as precisely the limits of sequences of eigenvalues of the finite-dimensional matrices.
\end{proof}

\begin{theorem}
	Consider $\frac{1}{2} \le \sigma \le \left(1 + \frac{\sqrt{2}}{2} \right) \theta$. Then there exists $c>0$, independent of $\Delta t\le 1$, $h_1$ and $h_2$, such that
	\begin{enumerate}
	\item
	\begin{equation}
	|T(\tilde{z}_0, \tilde{z}_1, \tilde{z}_2)| \le 1 + c \Delta t, \qquad \forall \phi_1, \phi_2 \in [0,2\pi],
	\end{equation}
	i.e., the scheme is von Neumann stable;
	\item
	\label{part2}
	\begin{equation}
	|U_n|_2 %:= \sum_{j=-\infty}^{\infty} U_{n}(j)^2 
	\le {\rm e}^{c n \Delta t} |U_0|_2, \qquad \forall n\ge 0,
	\end{equation}
	for $|U_n|_2 = h_1 h_2 \left(\sum_{j_1,j_2=-\infty}^\infty |U_n(j_1,j_2)|^2\right)^{\scriptsize 1/2}$, i.e., the scheme is $l_2$ stable.
	\end{enumerate}
\end{theorem}
\begin{proof}

%The first inequality is a weak version of Theorem 3.2 of  \cite{intHoutStability} and can be proved in the similar way by 
%adding the term $c \Delta t$.

First, we have that
\begin{eqnarray*}
|T({z}_0, \tilde{z}_1, \tilde{z}_2)| = \left|1 + 2 \frac{{z}_0 + \tilde{z}}{p} - \frac{{z}_0 + \tilde{z}}{p^2} + \sigma \frac{({z}_0 + \tilde{z})^2}{p^2}\right|
\le 1,
\end{eqnarray*}
where as before $p = (1 - \theta \tilde{z}_1) (1 - \theta \tilde{z}_2)$ and $\tilde{z} = \tilde{z}_1 + \tilde{z}_2$. 
This follows from Theorem \ref{theor_inthout} because $\lambda_1$, $\lambda_2$ and $\lambda_{12}$ are positive and therefore (\ref{karelineq}) is still satisfied with $z_1$ and $z_2$ replaced by $\tilde{z}_1$ and $\tilde{z}_2$.

We have
\begin{eqnarray*}
T(\tilde{z}_0, \tilde{z}_1, \tilde{z}_2) &=&
T({z}_0, \tilde{z}_1, \tilde{z}_2) +
2 \frac{s_0}{p} - \frac{s_0}{p^2} + \sigma \frac{2 s_0 (z_0 + \tilde{z}) + s_0^2}{p^2}.
\end{eqnarray*}
A simple calculation shows that $|s_0|\le c_0 \, \Delta t$ for a constant $c_0$ (independent of $\Delta t, h_1, h_2, \phi_1$, $\phi_2$;
indeed, $c_0=2 \lambda_1 + 2 \lambda_2 + 4 \lambda_{12}$
works for small enough $h_1$, $h_2$).
Therefore, and because $|p|\ge 1$, $|z_0 + \tilde{z}|/|p|\le c_1$ for a constant $c_1$,
\[
\left|2 \frac{s_0}{p} - \frac{s_0}{p^2} + \sigma \frac{2 s_0 (z_0 + \tilde{z}) + s_0^2}{p^2} \right|
\le c \Delta t,
\]
for any $c\ge (3 + 2 \sigma c_1 + c_0 \sigma) c_0$.
From this the first statement follows.

We can now deduce part \ref{part2} by a standard argument. For the discrete-continuous Fourier transform
\[
%\mathcal{F}: 
l_2(\mathbb{Z}^2) \rightarrow L_2(-\pi,\pi)^2, \qquad U \rightarrow \widehat{U}, \qquad \widehat{U}(\phi_1,\phi_2) = h_1 h_2 \sum_{j,k \in \mathbb{Z}} U(j,k) {\rm e}^{-i (j \phi_1 + k \phi_2)},
\]
we have
\[
\widehat{U}_{n+1}(\phi_1,\phi_2) = T(\tilde{z}_0, \tilde{z}_1, \tilde{z}_2) \, \widehat{U}_{n}(\phi_1,\phi_2), \qquad \forall n\ge 0.
\]
Then, by Parseval,
\begin{eqnarray*}
|U_n|_2^2 &=& \frac{1}{4 \pi^2} |\widehat{U}_n|^2  \\
&=& \frac{1}{4 \pi^2} 
\frac{1}{h_1^2 h_2^2} \int_{-\pi}^\pi |\widehat{U}_n(\phi_1,\phi_2)|^2 
\, {\rm d} \phi_1 \, {\rm d} \phi_2 \\
&\le& \frac{1}{4 \pi^2} 
\frac{1}{h_1^2 h_2^2} \int_{-\pi}^\pi (1+c \Delta t)^{2n} |\widehat{U}_0(\phi_1,\phi_2)|^2 
\, {\rm d} \phi_1 \, {\rm d} \phi_2 \\
&\le& {\rm e}^{2 c n \Delta t}  \frac{1}{4 \pi^2} 
\frac{1}{h_1^2 h_2^2} \int_{-\pi}^\pi |\widehat{U}_0(\phi_1,\phi_2)|^2 
\, {\rm d} \phi_1 \, {\rm d} \phi_2 \\ 
&=& {\rm e}^{2 c n \Delta t}  |U_0|_2^2.
\end{eqnarray*}
\end{proof}

%\begin{remark}
This ($l_2$-)stability result together with second order consistency implies ($l_2$-)convergence of second order for all solutions which
are sufficiently smooth that the truncation error is defined and bounded. In our setting, where the initial condition is discontinuous, this is not given. Since the step function lies in the ($l_2$-)closure of smooth functions, convergence is guaranteed, but usually not of second order. We show this empirically in the next section and demonstrate how second order convergence can be restored practically.
%\end{remark}

\subsection{Discontinuous boundary and terminal conditions}

It is well documented (see, e.g.\ \cite{pooley2003}) that the spatial convergence order of central finite difference schemes is generally reduced to one for discontinuous payoffs. Moreover, the time convergence order of the Crank-Nicolson scheme is reduced to one due to the lack of damping of high-frequency components of the error, and this behaviour is inherited by the HV scheme. We address these two issues in the following way.

First, we smooth the terminal condition by the method of local averaging from \cite{pooley2003}, i.e., instead of using nodal values of $\phi$ directly, we use the approximation
\[
\phi(x_1^i,x_2^j) \approx \frac{1}{h_1 h_2} \int_{x_2^i-h_2/2}^{x_2^i+h_2/2} \int_{x_1^j-h_1/2}^{x_1^j+h_1/2}
\phi(\xi_1,\xi_2) \, d\xi_1 d\xi_2.
\]
For step functions with values of 0 and 1, this procedure attaches to each node the fraction of the area where the payoff is 1, in a cell of of size $h
_1 \times h_2$ centred at this point.

We illustrate the convergence improvement on the example of joint survival probabilities. Other quantities show a similar behaviour.
The model parameters in the following tests are the same as in the next section, specifically Table \ref{table:params}.

We choose $\sigma = \frac{1}{2}$ and $\theta = \frac{3}{4}$ in the HV scheme. 

The observed convergence with and without this smoothing procedure is shown in Figure \ref{fig_conv1}.
We choose the $l_2$-norm for its closeness to the stability analysis -- in the periodic case, Fourier analysis gives convergence results in $l_2$ -- and the  $l_\infty$-norm for its relevance to the problem at hand, where we are interested in the solution pointwise.
The behaviour in the $l_1$-norm is very similar.

Hereby, for a method of order $p\ge 1$ we estimate the error by extrapolation as
\begin{equation*}
|Q^{nX}(x_1, x_2) - Q(x_1, x_2)| \approx \frac{1}{2^p-1}  |Q^{nX}(x_1, x_2) - Q^{nX/2}(x_1, x_2)|,
\end{equation*}
where $Q$ is the exact solution, $Q^{nX}$ the solution with $nX$ mesh points,
and the norms are computed by either taking the maximum over mesh points or numerical quadrature.
Here, $nT=1000$ is fixed.

\begin{figure}[H]
	\begin{center}
%				\subfloat[$l_1$-norm.]{\includegraphics[width=0.33\textwidth]{conv_analysis2.png}}
%				\subfloat[$l_2$-norm.]{\includegraphics[width=0.33\textwidth]{conv_analysis1.png}}
%				\subfloat[$l_{\infty}$-norm.]{\includegraphics[width=0.33\textwidth]{conv_analysis3.png}}\\
%				\subfloat[$l_1$-norm.]{\includegraphics[width=0.45\textwidth]{conv_analysis2.png}}
				\subfloat[$l_2$-norm.]{\includegraphics[width=0.49\textwidth]{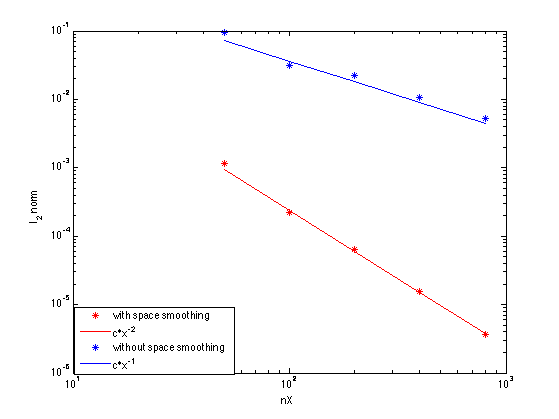}} \hfill
				\subfloat[$l_{\infty}$-norm.]{\includegraphics[width=0.49\textwidth]{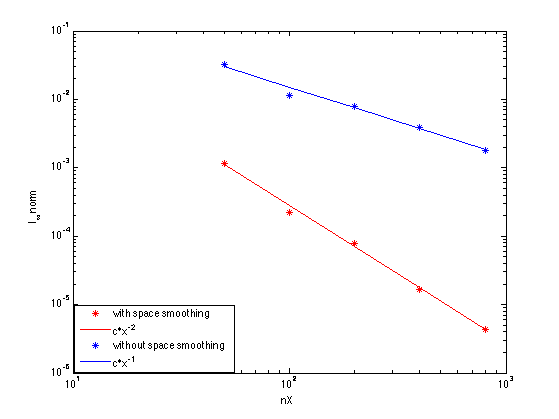}}\\
	\end{center}		
	\vspace{-20pt}
	\caption{Convergence analysis for $l_2$- and $l_{\infty}$-norms of the error depending on the mesh size with fixed time-step.}
 	\label{fig_conv1}
\end{figure}

The convergence is clearly of first order without averaging and of second order with averaging.

Second, we modify the scheme using the idea from \cite{reisinger2013} by changing the time variable $\tilde{t} = \sqrt{t}$. This change of variables 
leads to the new PDE
\[
		\frac{\partial V}{\partial \tilde{t}} + 2 \tilde{t} \mathcal{L} V = 2 \tau \chi(\tilde{t}^2, x),
\]
instead of (\ref{kolm_1}), to which we apply the numerical scheme. %, but improves the convergence rate. 
%From numerical results in Section \ref{numerical_experiments}, we can observe second order of convergence.

In Figure \ref{fig_conv2}, we show the convergence with and without time change, estimating the errors in a similar way to above, with $nX=800$ fixed.
 \begin{figure}[H]
	\begin{center}
%				\subfloat[$l_1$-norm.]{\includegraphics[width=0.33\textwidth]{conv_analysis5.png}}
%				\subfloat[$l_2$-norm.]{\includegraphics[width=0.33\textwidth]{conv_analysis4.png}}
%				\subfloat[$l_{\infty}$-norm.]{\includegraphics[width=0.33\textwidth]{conv_analysis6.png}}\\
%				\subfloat[$l_1$-norm.]{\includegraphics[width=0.45\textwidth]{conv_analysis5.png}}
				\subfloat[$l_2$-norm.]{\includegraphics[width=0.49\textwidth]{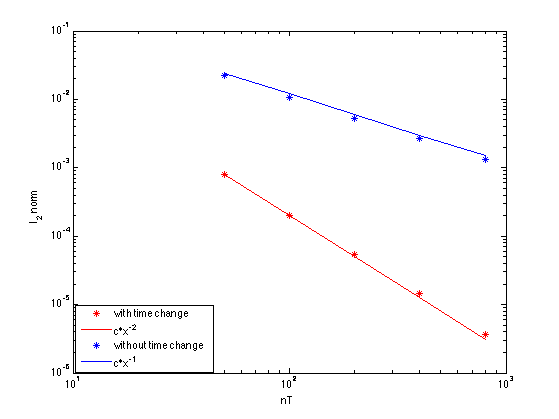}} \hfill
				\subfloat[$l_{\infty}$-norm.]{\includegraphics[width=0.49\textwidth]{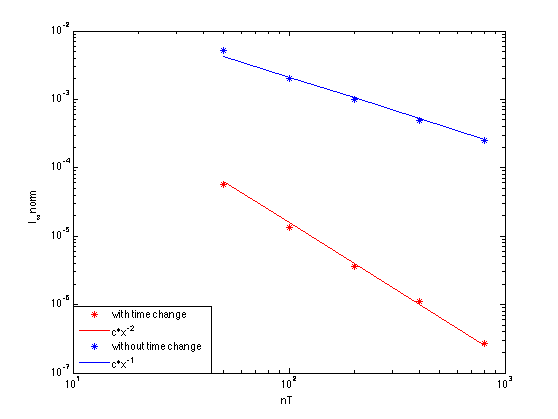}}\\

	\end{center}		
	\vspace{-20pt}
	\caption{Convergence analysis for $l_2$- and $l_{\infty}$-norms of the error depending on time-step with fixed mesh size.}
 	\label{fig_conv2}
\end{figure}

The convergence is clearly of first order without time change and of second order with time change. We took here $T=5$ to illustrate the effect more clearly.

%To justify second order convergence rate, in Figure \ref{fig_conv} we present the $l_2$-norm of error depending on the mesh size computed for joint survival probability. %We choose the number of time steps $n = 2m$, where $m$ is the mesh size in each direction.
% \begin{figure}[H]
%	\begin{center}
%		\includegraphics[width=0.9\textwidth]{conv_analysis.png}
%	\end{center}
%	\vspace{-20pt}
%	\caption{$l_2$ norm of error depending on the mesh size}
% 	\label{fig_conv}
%\end{figure}
%

\section{Numerical experiments}
\label{numerical_experiments}
In this section, we analyze the model characteristics and the impact of jumps. Specifically, we compute joint and marginal survival probabilities, CDS and FTD spreads as well as CVA and DVA depending on initial asset values. We also compute the difference between the solution with and without jumps.

Consider the parameters in Table \ref{table:params}.
\begin{table}[H]
	\begin{center}
		\begin{tabular}{| c | c | c | c | c | c | c | c | c | c | c | c |}
			\hline
			$L_{1,0}$ & $L_{2, 0}$ & $L_{12, 0}$ & $L_{21, 0}$ & $R_1$ & $R_2$ & $T$ & $\sigma_1$ & $\sigma_2$ & $\rho$ & $\varsigma_1$ & $\varsigma_2$ \\ 
			\hline
			60 & 70 & 10 & 15 & 0.4 & 0.45 & 1 & 1 & 1  & 0.5 & 1 & 1 \\
			\hline
		\end{tabular}
	\caption{Model parameters.\label{table:params}}		
	\end{center}
\end{table}
For the model with jumps, we further consider the parameters in Table \ref{table:jumps}.
\begin{table}[H]
	\begin{center}
		\begin{tabular}{| c | c | c | }
			\hline
			 $\lambda_1$& $\lambda_2$ & $\lambda_{12}$ \\
			\hline
			0.5 & 0.5 & 0.3 \\
			\hline
		\end{tabular}
		\caption{Jump intensities.\label{table:jumps}}
	\end{center}
\end{table}

We compute all tests using a $100\times100$ spatial grid with the maximum values $X_1^{100} = X_2^{100} = 10$ and constant time step $\Delta \tau = 0.01$. As the parameters of the HV scheme, we choose $\sigma = \frac{1}{2}$ and $\theta = \frac{3}{4}$. 

In Figures \ref{jointSurvProb1}--\ref{CVA1} we present various model characteristics and compare the results with and without jumps. From these figures, we can observe that jumps can have a significant impact, especially near the default boundaries:

\begin{itemize}
\item
in Figure \ref{jointSurvProb1} for the joint survival probability,
the biggest impact of jumps is around the default boundaries for both $x_1$ and $x_2$;
\item
in Figure \ref{marginalSurvProb1} for the marginal survival probability of the first bank,
we can observe that the biggest impact of jumps is near the default boundary of the first bank;
\item
for the CDS spread, in Figure \ref{CDSPrice1}, (b), the biggest impact of jumps is also seen near the default boundary, but it has the opposite direction, because jumps can only increase the CDS spread;
\item
in Figure \ref{FTDPrice1}, (d) for FTD the spread, the biggest impact of jumps is near both default boundaries, and it has a positive impact;
\item
finally, for CVA, (f), the highest impact of jumps is near the default boundary of the first bank, see Figure \ref{CVA1}.
\end{itemize}

\begin{figure}%[H]
	\begin{center}
		\subfloat[]{\includegraphics[width=0.5\textwidth]{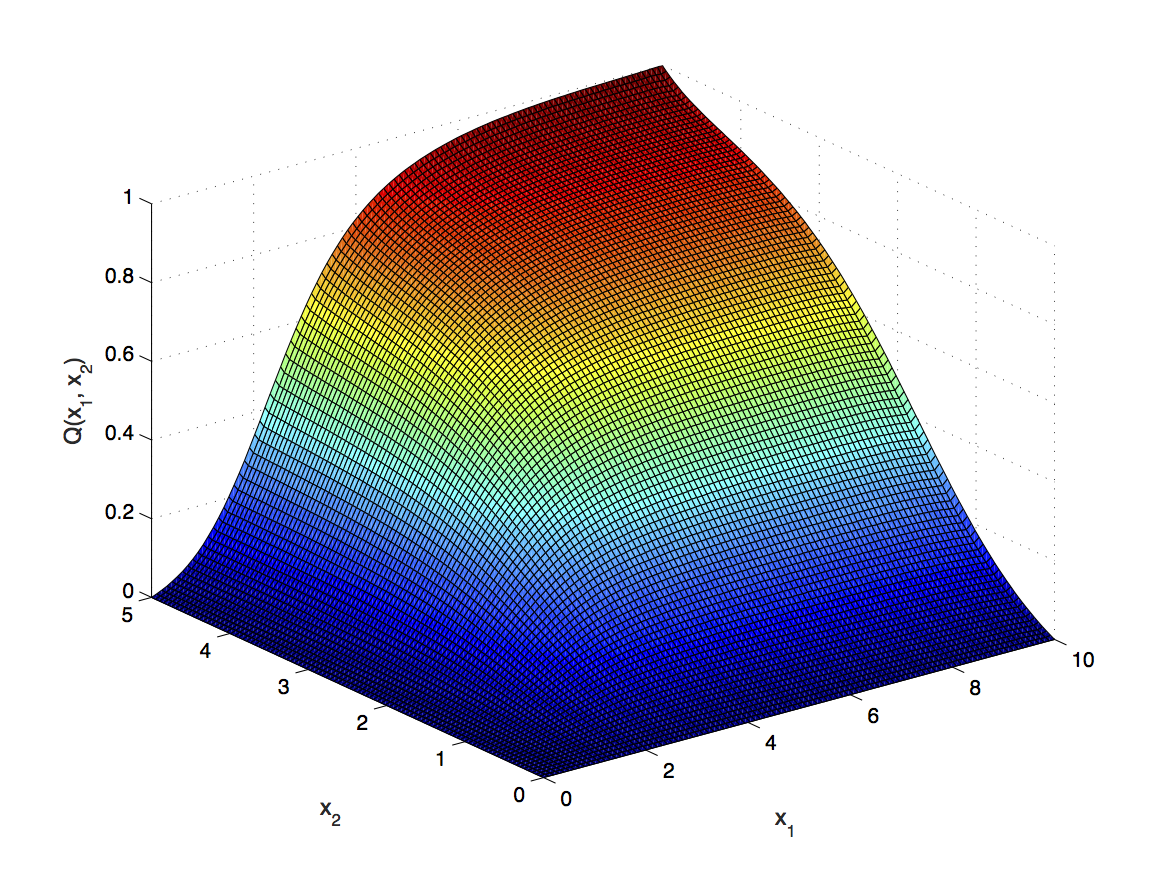}}
		\subfloat[]{\includegraphics[width=0.5\textwidth]{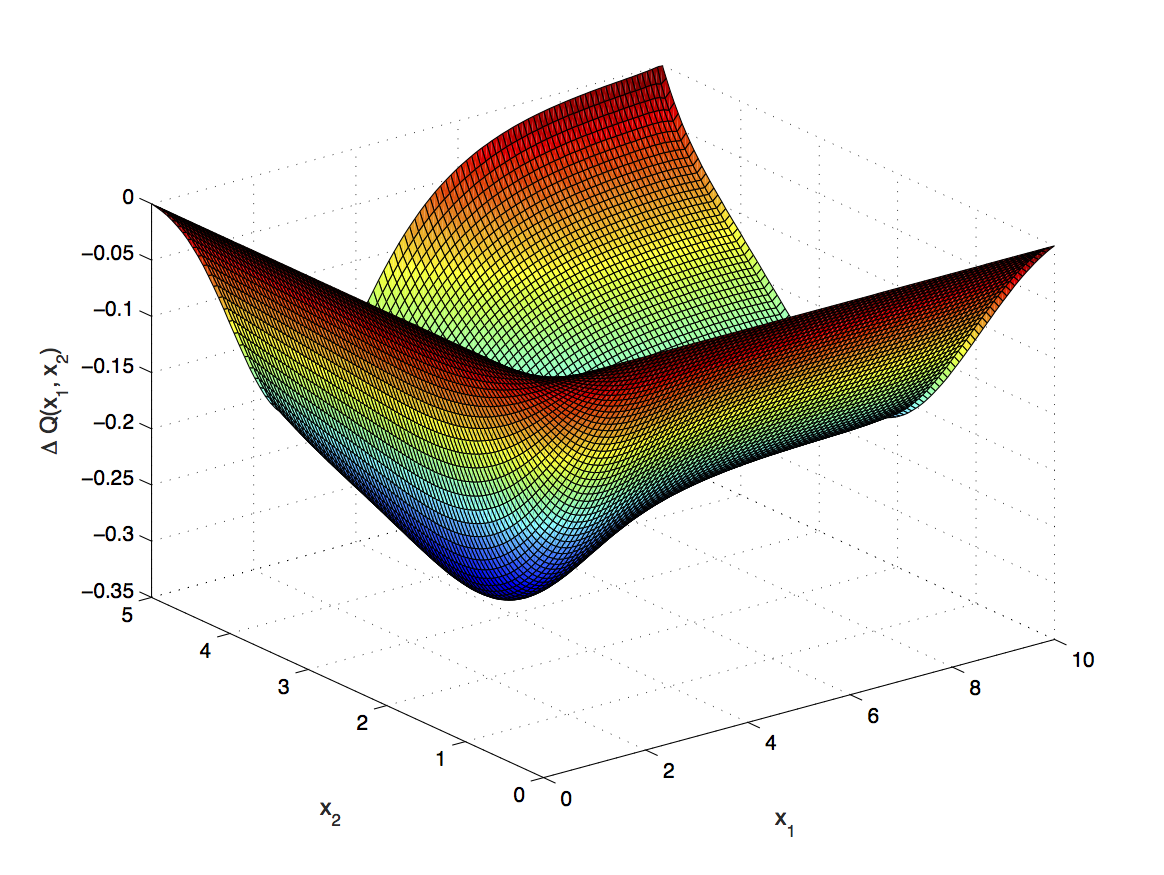}}\\
	\end{center}
	\vspace{-20pt}
	\caption{The joint survival probability: (a) value, (b) difference between model with and without jumps.}
 	\label{jointSurvProb1}
\end{figure}

 \begin{figure}%[H]
	\begin{center}
		\subfloat[]{\includegraphics[width=0.5\textwidth]{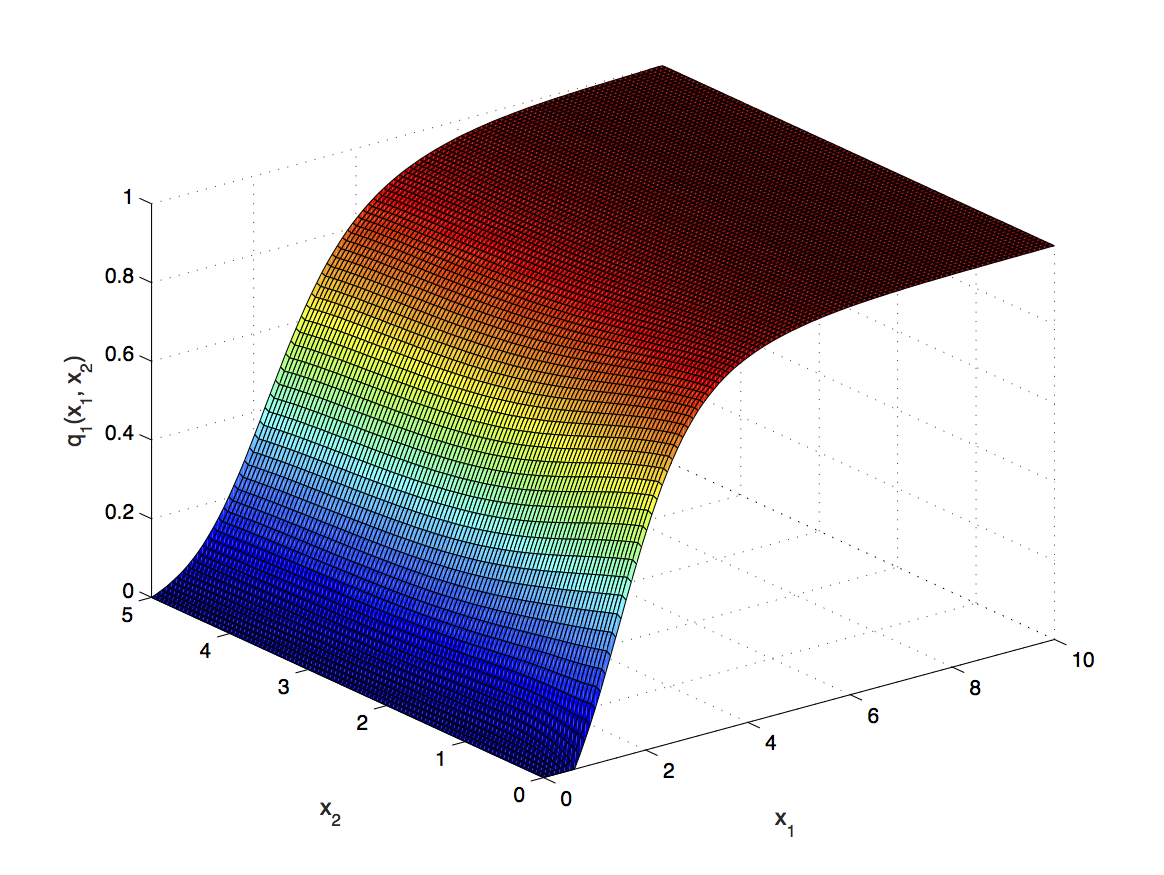}}
		\subfloat[]{\includegraphics[width=0.5\textwidth]{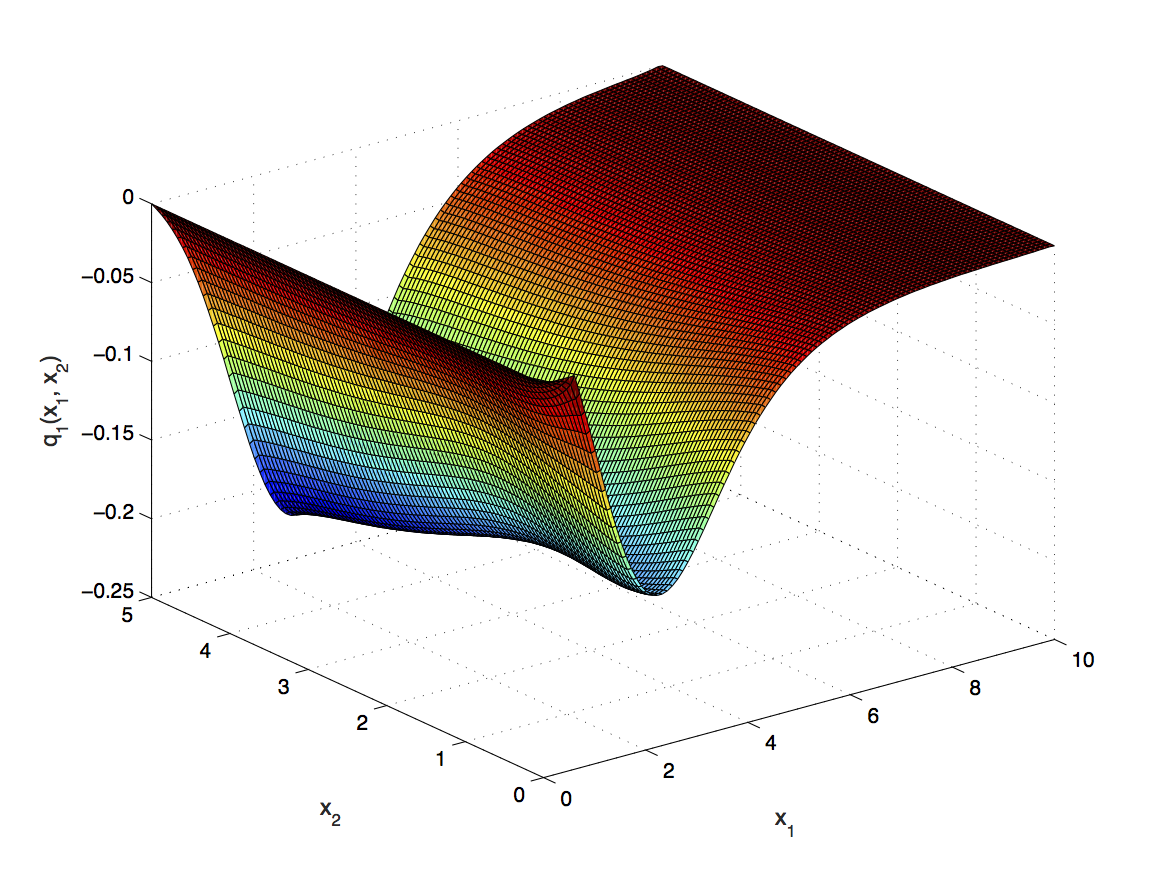}}\\
	\end{center}
	\vspace{-20pt}
	\caption{The marginal survival probability: (a) value, (b)  difference between model with and without jumps.}
 	\label{marginalSurvProb1}
\end{figure}

 \begin{figure}%[H]
	\begin{center}
		\subfloat[]{\includegraphics[width=0.49\textwidth]{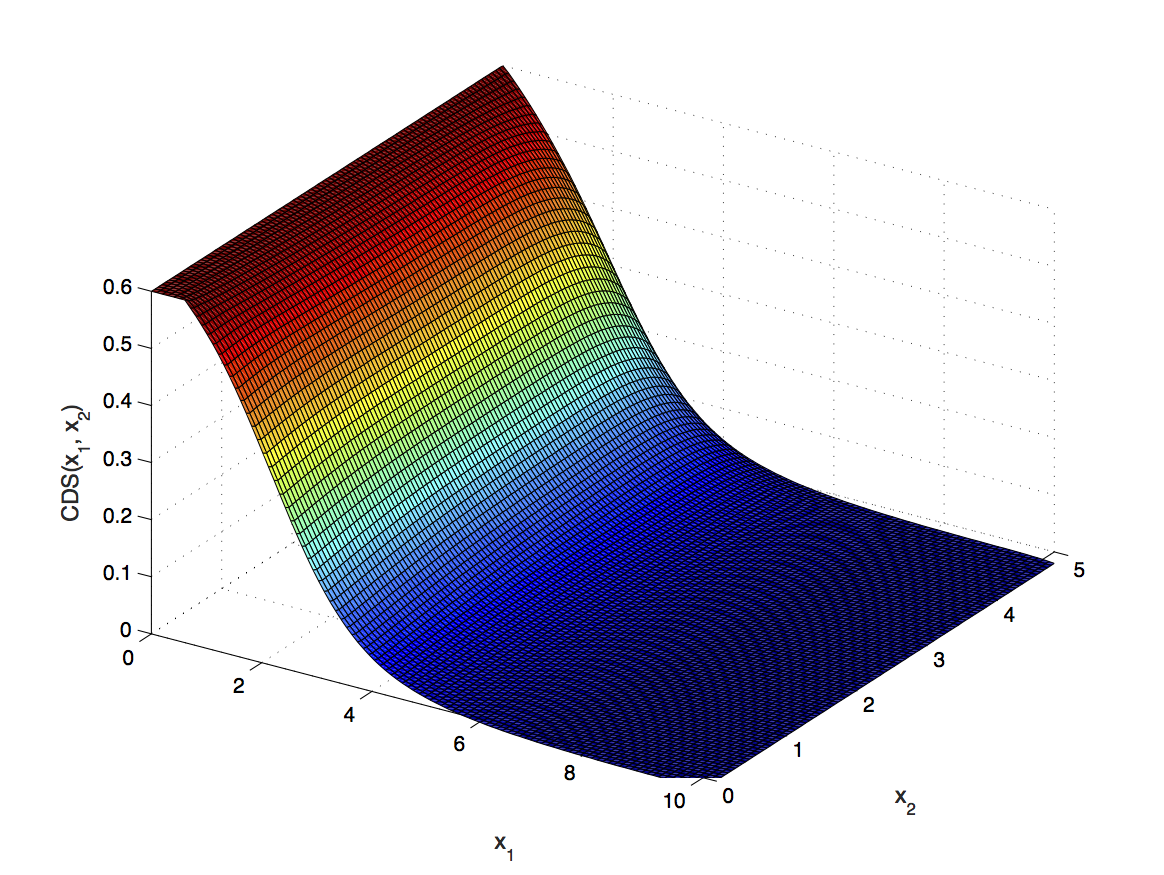}}
		\subfloat[]{\includegraphics[width=0.49\textwidth]{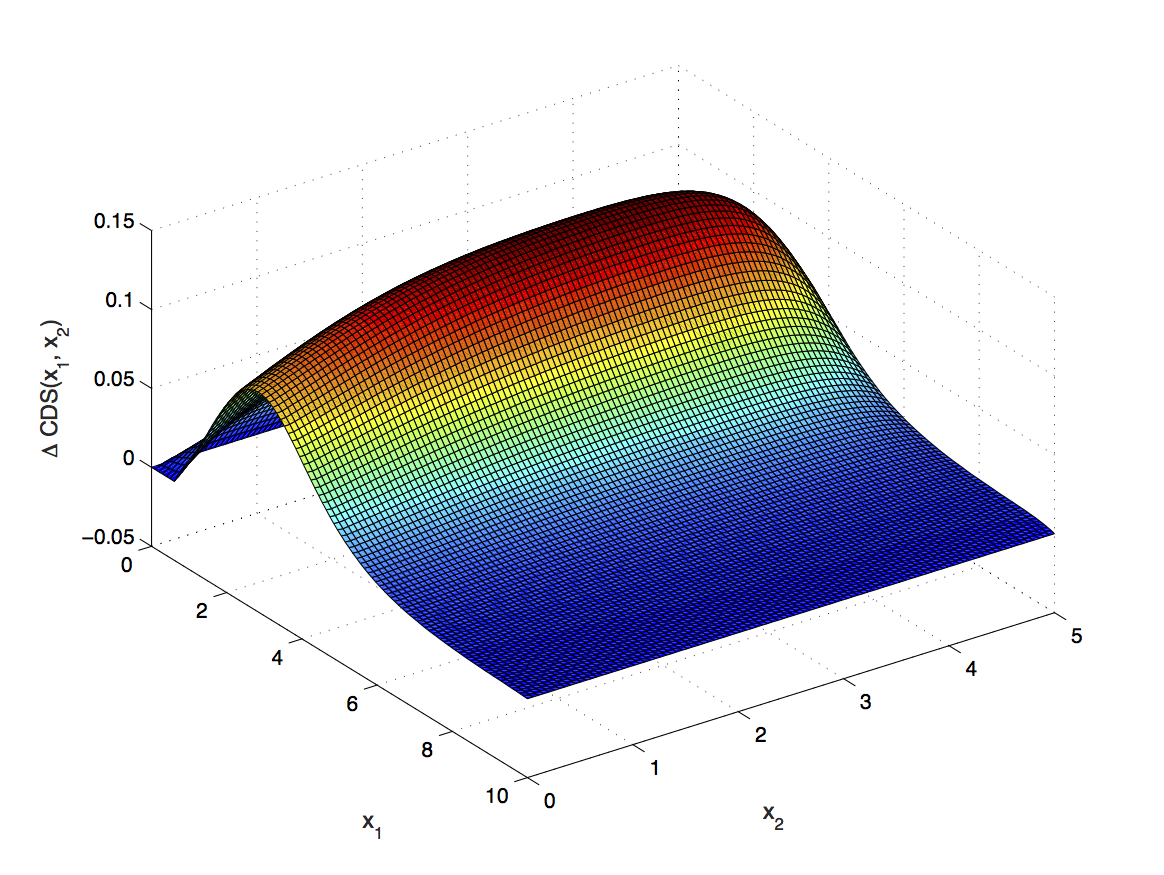}}\\
		\subfloat[]{\includegraphics[width=0.49\textwidth]{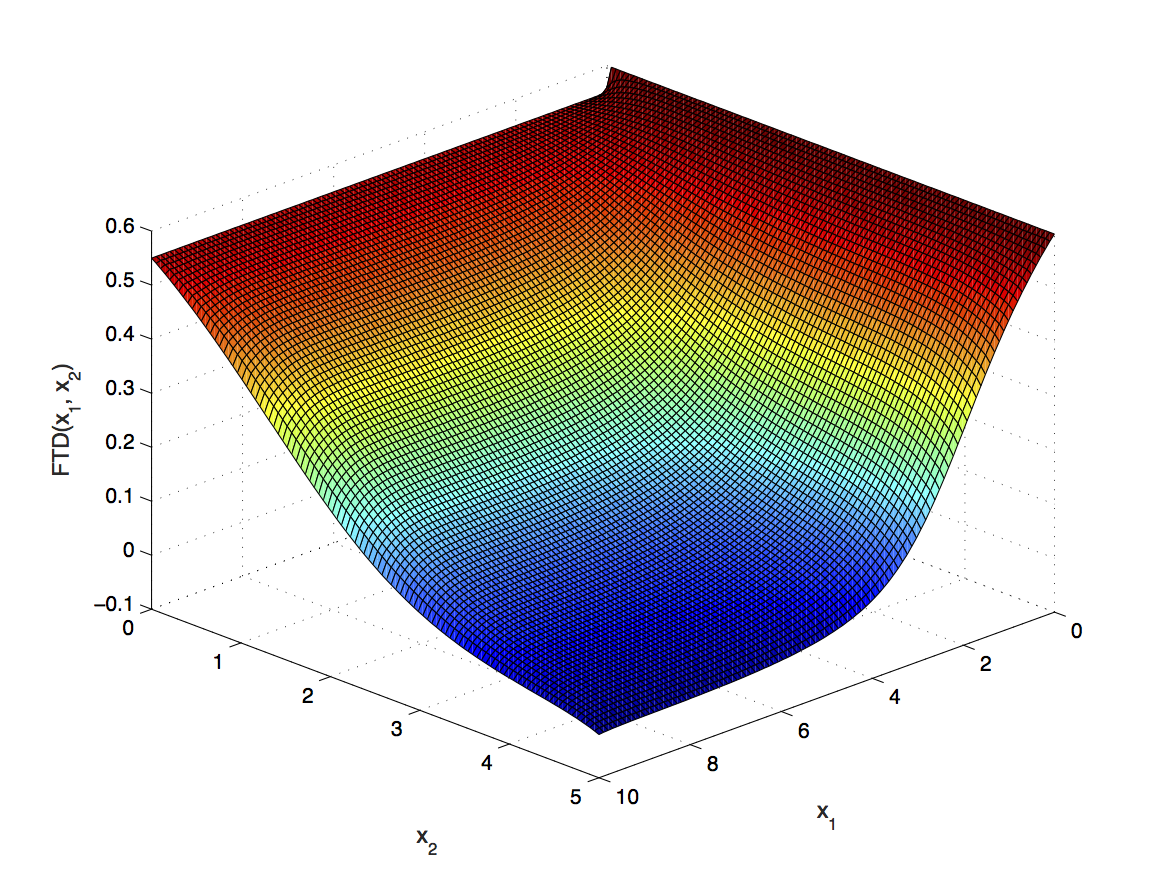}}
		\subfloat[]{\includegraphics[width=0.49\textwidth]{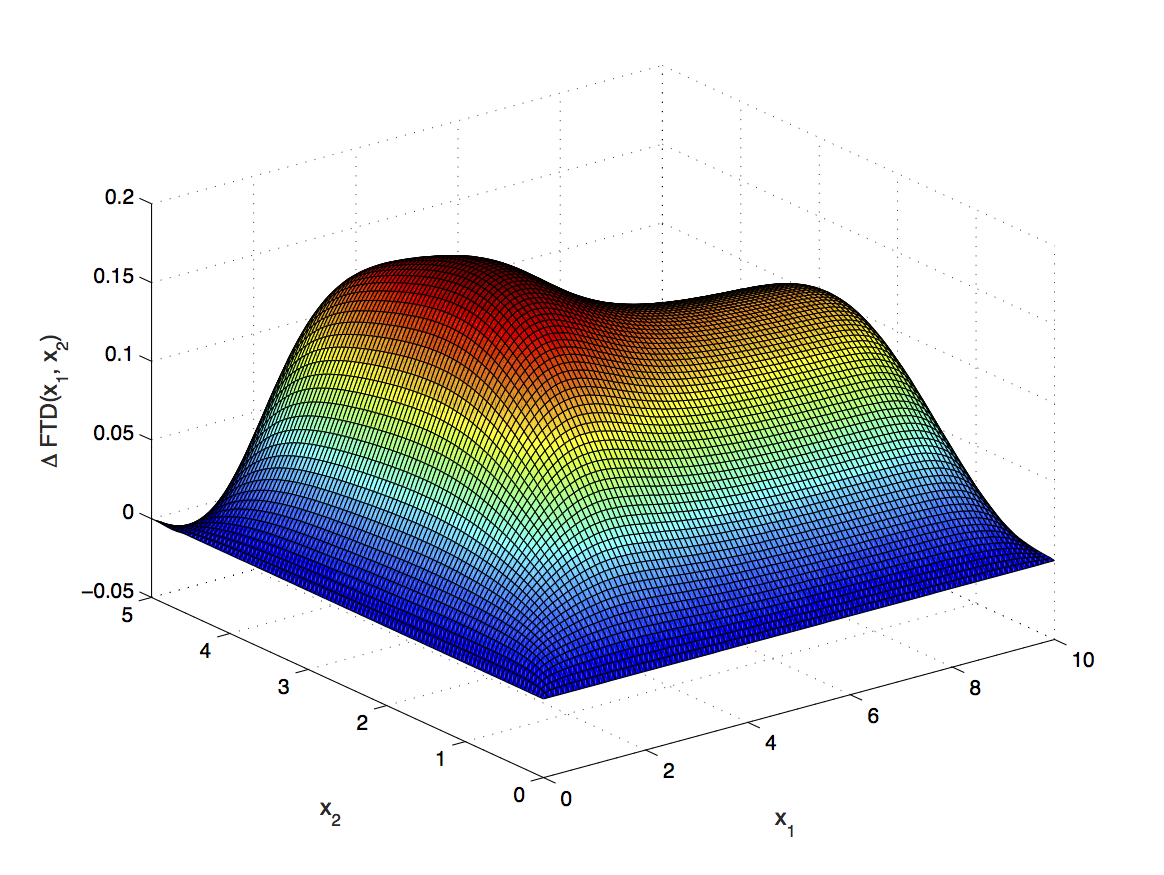}}\\
		\subfloat[]{\includegraphics[width=0.49\textwidth]{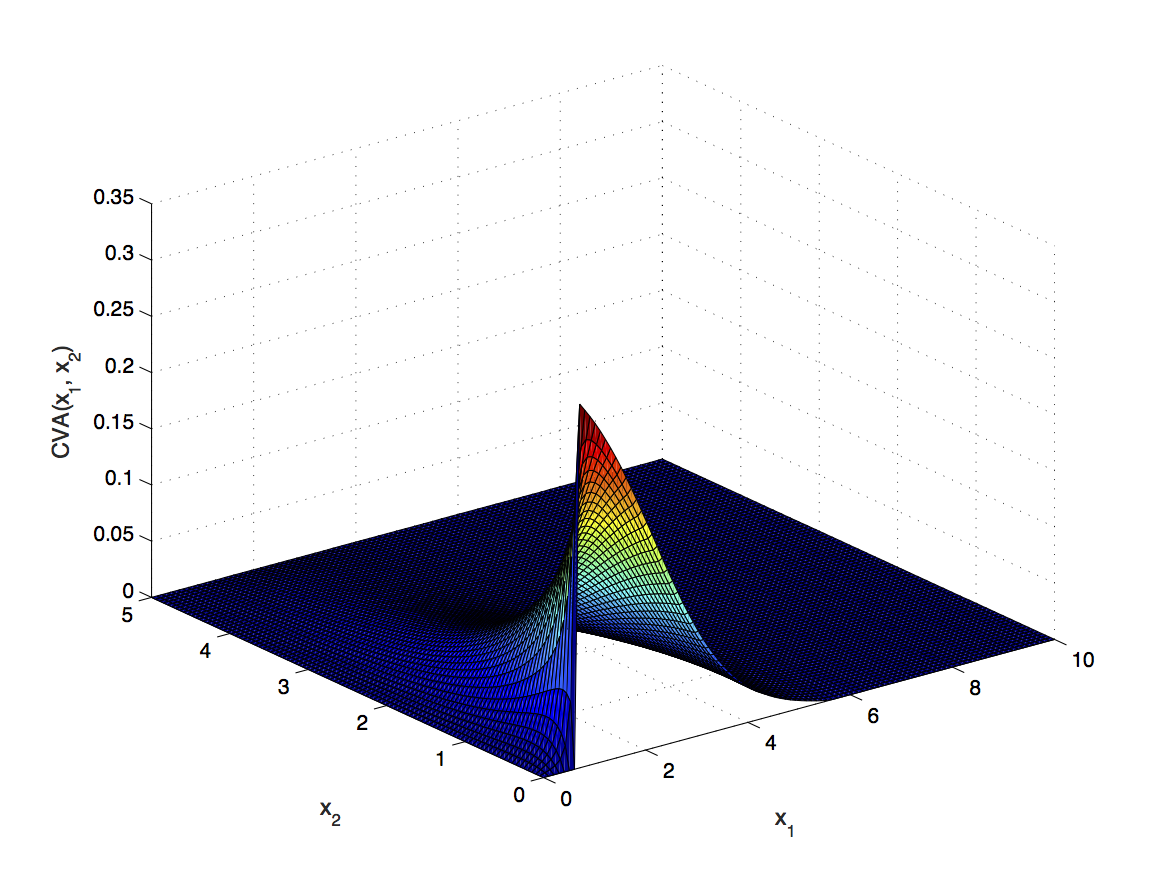}}
		\subfloat[]{\includegraphics[width=0.49\textwidth]{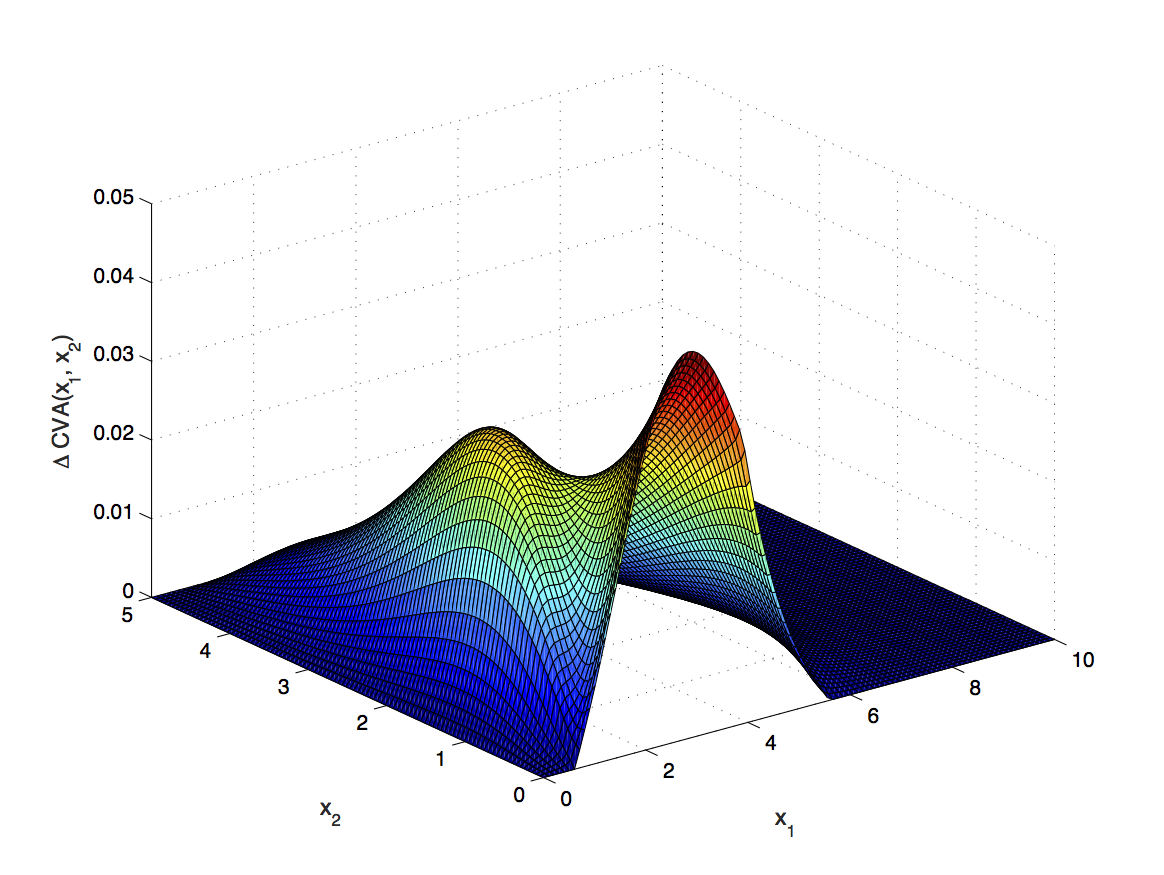}}
	\end{center}
	\vspace{-20pt}
	\caption{
	Values of different credit products with left the value and right the difference between model with and without jumps.
	Top row: Credit Default Swap spread, written on the first bank.
	Middle row: First-to-Default spread.
	Bottom row: CVA of CDS, where the first bank is Reference name (RN) and the second bank is Protection Seller (PS).
	}
 	\label{CDSPrice1}
	\label{FTDPrice1}
	\label{CVA1}
\end{figure}

% \begin{figure}[H]
%	\begin{center}
%		\subfloat[]{\includegraphics[width=0.55\textwidth]{cds_jumps.png}}
%		\subfloat[]{\includegraphics[width=0.55\textwidth]{cds_diff_jumps.png}}\\
%	\end{center}
%	\vspace{-20pt}
%	\caption{Credit Default Swap spread, written on the first bank: (a) value, (b)  difference between model with and without jumps.}
% 	\label{CDSPrice1}
%\end{figure}
%
% \begin{figure}[H]
%	\begin{center}
%		\subfloat[]{\includegraphics[width=0.55\textwidth]{ftd_jumps.png}}
%		\subfloat[]{\includegraphics[width=0.55\textwidth]{ftd_diff_jumps.png}}\\
%	\end{center}
%	\vspace{-20pt}
%	\caption{First-to-Default spread: (a) value, (b)  difference between model with and without jumps.}
% 	\label{FTDPrice1}
%\end{figure}
%
% \begin{figure}[H]
%	\begin{center}
%		\subfloat[]{\includegraphics[width=0.55\textwidth]{cva_jumps.png}}
%		\subfloat[]{\includegraphics[width=0.55\textwidth]{cva_diff_jumps.png}}\\
%	\end{center}
%	\vspace{-20pt}
%	\caption{CVA of CDS, where the first bank is Reference name (RN) and the second bank is Protection Seller (PS) : (a) value, (b)  difference between model with and without jumps.}
% 	\label{CVA1}
%\end{figure}

%From the figures above, we can see that jumps have the most impact near the default boundaries, especially with large intensities. This might have a significant impact in model characteristics.

\section{Calibration}

In this section we present calibration results of the model. There are eight unknown parameters, see \eqref{kolmogorov_backward}--\eqref{j12_eq}: $\sigma_1, \sigma_2, \rho, \varsigma_1, \varsigma_2, \lambda_1, \lambda_2, \lambda_{12}$. We use CDS and equity put option prices (with different strikes) as market data. If FTD contracts are available, one can use them to estimate $\rho$ and $\lambda_{12}$. Otherwise, historical estimation with share prices time series can be used. 

The data for external liabilities can be found in banks' balance sheets, which are publicly available. Usually, mutual liabilities data are not public information, thus we made an assumption that they are a fixed proportion of the total liabilities, which coincides with \cite{DavidLehar}. In particular, we fix the mutual liabilities as 5\% of total liabilities.

The asset's value is the sum of liabilities and equity price.

We choose Unicredit Bank as the first bank and Santander as the second bank. In Table \ref{data_table} we provide their equity price $E_i$, assets $A_i$ and liabilities $L_i$. As in \cite{LiptonSepp}, the liabilities are computed as a ratio of total liabilities and shares outstanding.

\begin{table}[H]
	\begin{center}
		\begin{tabular}{| c | c | c | c | c | c |}
			\hline
			$E_1(0)$ & $L_1(0)$ & $A_1(0)$ & $E_2(0)$ & $L_2(0)$ & $A_2(0)$  \\ 
			\hline
			6.02&  137.70& 143.72& 6.23 & 86.41 & 92.64 \\
			\hline
		\end{tabular}
		\caption{Assets and liabilities on 30/06/2015 (Bloomberg).}
		\label{data_table}	
	\end{center}
\end{table}
For the calibration we choose 1-year at-the-money, in-the-money, and out-of-the-money equity put options on the banks, and 1-year CDS contracts. Since the spreads of CDS are usually significantly lower than the option prices, we scale them by some weight $w_i$ in the objective function. As a result, we have the following 6-dimensional minimization problem:
\begin{multline}
	\label{calibration_eq}
	\min_{\theta} \{ w_1 (V^{CDS}_1(\theta) - \bar{V}^{CDS}_1)^2 + \sum \limits_{i = 1}^3 (V^{opt}_1(K_{i, 1}, \theta) - \bar{V}^{opt}_1(K_{i, 1}))^2 + \\
	+ w_2 (V^{CDS}_2(\theta) - \bar{V}^{CDS}_2)^2  + \sum \limits_{i = 1}^3 (V^{opt}_2(K_{i, 2}, \theta) - \bar{V}^{opt}_2(K_{i, 2}))^2  \},
\end{multline}
 where $\theta = (\sigma_1, \sigma_2, \lambda_1, \lambda_2, \varsigma_1, \varsigma_2)$, $V^{CDS}_i(\theta)$ is the model CDS spread on the $i$-th bank and $\bar{V}^{CDS}_i$ is the market CDS spread on the $i$-th bank, $V^{opt}_1(K, \theta)$ is the model price of the equity put option on the $i$-th bank with the strike $K$ and $\bar{V}^{opt}_i(K)$ is the market price of the equity put option on the $i$-th bank with strike $K$. Strikes $K_{1, j}, K_{2, j}$, and $K_{3, j}$ are chosen in such a way to take into account the smile. In particular, we choose $K_{1, j} = 1.1 E_j, K_{2, j} = E_j, K_{3, j} = 0.9 E_j$.

In order to find the global minimum of \eqref{calibration_eq} by a Newton-type method, we need to find a good starting point, otherwise an optmization procedure might finish in local minima which are not global minima. To choose the starting point, we calibrate one-dimensional models for each bank without mutual liabilities
\begin{multline}
	\label{calibration_eq1d}
	\min_{\theta_j} \{ w_j (V^{CDS}_j(\theta_j) - \bar{V}^{CDS}_j)^2 + (V^{opt}_j(K_{1, j}, \theta_j) - \bar{V}^{opt}_j(K_{1, j}))^2 + \\
	+(V^{opt}_j(K_{2, j}, \theta_j) - \bar{V}^{opt}_i(K_{2, j}))^2  + (V^{opt}_j(K_{3, j}, \theta_j) - \bar{V}^{opt}_j(K_{3, j}))^2  \},
\end{multline}
where $\theta_j = (\sigma_j, \lambda_j, \varsigma_j)$ for $j = 1, 2$.

The global minima of \eqref{calibration_eq1d} can be found via the {\bf{chebfun toolbox}} (\cite{Trefethen}) that uses Chebyshev polynomials to approximate the function, and then the global minima can be easily found. The calibration results of the one-dimensional model for the first and the second banks are presented in Table \ref{table:params_1d}. We note that the global minima of \eqref{calibration_eq} cannot be found via the chebfun toolbox, since it works with functions up to three variables. There are also more fundamental complexity issues for higher-dimensional tensor product interpolation.

\begin{table}[H]
	\begin{center}
		\begin{tabular}{| c | c | c | c | c | c |}
			\hline
			$\sigma_1$ & $\lambda_1$ & $\varsigma_1$ & $\sigma_2$ & $\lambda_2$ & $\varsigma_2$  \\ 
			\hline
			 0.0117&  0.1001& 0.3661& 0.0154 & 0.0160 & 0.0545\\
			\hline
		\end{tabular}
		\caption{Calibrated parameters of one-dimensional models on 30/06/2015 for $T = 1$.}
		\label{table:params_1d}	
	\end{center}
\end{table}

Similar to \cite{LiptonSepp}, for simplicity, we further assume that 
\begin{equation}
	\lambda_{\{12\}} = \rho \cdot \min(\lambda_1, \lambda_2).
	\label{lambda_assumption}
\end{equation}
 Then, we estimate $\rho$ from historical data. We take one year daily  equity prices $E_i(t)$ by time series (from Bloomberg) and estimate the covariance of asset returns $r_t^i = \frac{\Delta A_i(t)}{A_i(t)}$
\begin{equation}
	\widehat{\cov}(A_1, A_2) = \sum \limits_{i = 1}^n \left(r_{i, 1} - \bar{r_1} \right)\left(r_{i, 2} - \bar{r}_2  \right),
	\label{cov_est}
\end{equation}
where $\bar{r}_1$ and $\bar{r}_2$ are the sample mean of asset returns.

Using \eqref{assets_dynamics}, we can see that \eqref{cov_est} converges to
\begin{equation}
	\widehat{\cov}(A_1, A_2) \underset{n \to +\infty}{\longrightarrow} \sigma_1 \sigma_2 \left( \rho+ \lambda_{\{12\}} /(\varsigma_1 \varsigma_2) \right).
\end{equation}
Using the last equation and \eqref{lambda_assumption}, we can extract the estimated values of $\rho$ and $\lambda_{\{12\}}$. The estimation results are in Table \ref{table:corr_params}.
\begin{table}[H]
	\begin{center}
		\begin{tabular}{| c | c | c | }
			\hline
			& $\rho$ & $\lambda_{\{12\}} $ \\
			\hline
			Estimated value & 0.510 & 0.0188 \\
			\hline
			Confidence interval \footnotemark & (0.500, 0.526)& (0.0182, 0.0194) \\
			\hline
		\end{tabular}
		\caption{Historically estimated correlation coefficients on 30/06/2015 with 1 year window.}
		\label{table:corr_params}	
	\end{center}
\end{table}
\footnotetext{We use a $3 \sigma$ confidence interval.}

Finally, we perform a six-dimensional (constrained) optimization of \eqref{calibration_eq} with the starting point from Table \ref{table:params_1d} and correlation parameters from Table \ref{table:corr_params}. We choose different alternatives of mutual liabilities to have a clear picture how mutual liabilities influence on model parameters. We use the {\bf lsqnonlin} method in Matlab that uses a Trust Region Reflective algorithm \cite{conn2000trust} (with the gradient computed numerically). The model CDS spreads are computed using the method in Section \ref{CDS_pricing}, while equity option prices are computed in the usual finite-difference manner (see \cite{LiptonSepp} for details).  Results are presented in Table \ref{table:params_2d}.
%\begin{table}[H]
%	\begin{center}
%		\begin{tabular}{|c | c | c | c | c | c | c | c |}
%			\hline
%			$L_{12}$ & $L_{21}$ & $\sigma_1$ & $\lambda_1$ & $\varsigma_1$ & $\sigma_2$ & $\lambda_2$ & $\varsigma_2$  \\ 
%			\hline
%			0.0 & 0.0 & 0.0117&  0.1001& 0.3661& 0.0154 & 0.0160 & 0.0545 \\
%			2.0 & 3.0 & 0.0119 & 0.1012 & 0.3968 & 0.0153 & 0.0153 & 0.0517 \\
%			3.0 & 2.0 & 0.0119 & 0.0976 & 0.3841 & 0.0156 & 0.0154 & 0.0522 \\
%			5.0 & 6.0 & 0.0122 & 0.1021 & 0.4233 & 0.0154 & 0.0149 & 0.0491 \\
%			5.0 & 4.0 & 0.0120 & 0.1079 & 0.4212 & 0.0155 & 0.0149 & 0.0496 \\
%			5.0 & 0.0 & 0.0117 & 0.0989 & 0.3627 & 0.0160 & 0.0151 & 0.0527 \\
%			0.0 & 4.0 & 0.0117 & 0.0993 & 0.3796 & 0.0154 & 0.0145 & 0.0522 \\
%			\hline
%		\end{tabular}
%		\caption{Calibrated parameters of two-dimensional model with mutual liabilities on 30/06/2015 for $T = 1$.}
%		\label{table:params_2d}	
%	\end{center}
%\end{table}

\begin{table}[H]
	\begin{center}
		\begin{tabular}{|c | c | c | c | c | c | c | c |}
			\hline
			 Model & $\sigma_1$ & $\lambda_1$ & $\varsigma_1$ & $\sigma_2$ & $\lambda_2$ & $\varsigma_2$  \\ 
			\hline
			With jumps & 0.0122&  0.0950& 0.3958& 0.0160 & 0.0148 & 0.0505 \\
			Without jumps & 0.0206 & -- & -- & 0.0317 & -- & -- \\
			\hline
		\end{tabular}
		\caption{Calibrated parameters of two-dimensional model with mutual liabilities on 30/06/2015 for $T = 1$.}
		\label{table:params_2d}	
	\end{center}
\end{table}

In Table \ref{table:results} we present joint and marginal survival probabilities computed using the equations from Section \ref{section:joint}. From these results, we can conclude that jumps play an important role in the model.
\begin{table}[H]
	\begin{center}
		\begin{tabular}{|c | c | c | c | c | c | c | c |}
			\hline
			Model &Joint s/p & Marginal s/p   \\ 
			\hline
			With jumps & 0.9328 & 0.9666 \\
			Without jumps & 0.9717 & 0.9801 \\
			\hline
		\end{tabular}
		\caption{Joint and marginal survival probabilities for the calibrated models.}
		\label{table:results}	
	\end{center}
\end{table}

%\begin{table}[H]
%	\begin{center}
%		\begin{tabular}{|c | c | c | c | c | c | c | c |}
%			\hline
%			Model &Joint s/p & Marginal s/p   \\ 
%			\hline
%			0.0 & 0.0 & 0.8879 \\
%			2.0 & 3.0 &  0.8869 \\
%			3.0 & 2.0 &  0.8868\\
%			5.0 & 6.0 &  0.8861 \\
%			5.0 & 4.0 &  0.8810\\
%			5.0 & 0.0 &  0.8869\\
%			0.0 & 4.0 &  0.8866\\
%			\hline
%		\end{tabular}
%		\caption{Marginal survival probabilities for the calibrated models.}
%	\end{center}
%\end{table}

%
%K>> resid
%
%resid =
%
%    0.0167
%   -0.2809
%   -0.6476
%    1.0280
%
%K>> result
%
%result =
%
%    0.0180
%    0.0366
%    4.2949

%result =
%
%    0.0119
%    0.1089
%   30.5082
%
%K>> resid
%
%resid =
%
%   -0.0177
%    1.0682
%    0.6890
%   -1.6769

\section{Conclusion}
In this paper we considered a structural default model of interlinkage in the banking system. In particular, we studied a simplified setting of two banks numerically. This paper contains several new results. First, we developed a finite-difference method,
an extension of the Hundsdorfer-Verwer scheme, for the resulting partial integro-differential equation (PIDE), studied its stability and consistency. To deal with the integral component, we used the idea of its iterative computation from \cite{LiptonSepp}. The method gives second order convergence in both time and space variables and is unconditionally stable.

Second, by applying the finite-difference method, we computed various model characteristics, such as joint and marginal survival probabilities, CDS and FTD spreads, as well as CVA and DVA, and estimated the impact of jumps on the results. For a more sophisticated analysis, we calibrated the model to the market, and demonstrated a sizeable impact of jumps on joint and marginal survival probabilities in the case of two banks.
The development of numerical methods which are feasible for larger systems of banks appears to be an important future research direction.

From a numerical analysis perspective, we have extended the stability analysis of \cite{intHoutStability} to include an integral term arising from a jump-diffusion process with one-sided exponential jump size distribution.
By Fourier analysis, we were able to show that the scheme is stable in the $l_2$-sense when considering probability densities on an infinite domain. An interesting open question is the stability analysis in the presence of absorbing boundary conditions, such that the individual matrices involved in the splitting do not commute and the eigenvectors and eigenvalues of the combined operator cannot directly be computed.
We are planning to address this in future research.
\newpage
\appendix
\section{Pricing equations}
\subsection{Credit default swap}
\label{CDS_pricing}
A credit default swap (CDS) is a contract designed to exchange credit risk of a Reference Name (RN) between a Protection Buyer (PB) and a Protection Seller (PS). PB makes periodic coupon payments to PS conditional on no default of RN, up to the nearest payment date, in the exchange for receiving from PS the loss given RN's default.

Consider a CDS contract written on the first bank (RN), denote its price $C_1(t, x)$.\footnote{For the CDS contracts written on the second bank, the similar expression could be provided by analogy.} We assume that the coupon is paid continuously and equals to $c$. Then, the value of a standard CDS contract can be given (\cite{BieleckiRutkowski}) by the solution of  (\ref{kolm_1})--(\ref{kolm_2})  with $\chi(t, x) = c$ and terminal condition
\begin{equation*}
	\psi(x) = 
	\begin{cases}
		1 - \min(R_1, \tilde{R}_1(1)), \quad (x_1, x_2) \in D_2, \\
		1 - \min(R_1, \tilde{R}_1(\omega_2)), \quad (x_1, x_2) \in D_{12}, \\		
	\end{cases}
\end{equation*}
where $\omega_2 = \omega_2(x)$ is defined in (\ref{term_cond}) and 
\begin{equation*}
	\tilde{R}_1(\omega_2) = \min \left[1, \frac{A_1(T) +  \omega_2 L_{2 1}(T)}{L_1(T) + \omega_2 L_{12}(T)}\right].
\end{equation*}
Thus, the pricing problem for CDS contract on the first bank is
\begin{equation}
\begin{aligned}
		& \frac{\partial}{\partial t} C_1(t, x) + \mathcal{L} C_1(t, x) = c, \\
		& C_1(t, 0, x_2) = 1 - R_1, \quad C_1(t, \infty, x_2) = -c(T-t), \\
		& C_1(t, x_1, 0) = \Xi(t, x_1) = 
		\begin{cases}
			c_{1,0}(t, x_1), & x_1 \ge \tilde{\mu}_1, \\
			1-R_1, & x_1 < \tilde{\mu}_i,
		\end{cases} \quad C_1(t, x_1, \infty) = c_{1,\infty}(t, x_1),\\
		& C_1(T, x) = \psi(x) = 
	\begin{cases}
		1 - \min(R_1, \tilde{R}_1(1)), \quad (x_1, x_2) \in D_2, \\
		1 - \min(R_1, \tilde{R}_1(\omega_2)), \quad (x_1, x_2) \in D_{12}, \\		
	\end{cases}
\end{aligned}
\end{equation}
where $c_{1,0}(t, x_1)$ is the solution of the following boundary value problem:
\begin{equation}
\begin{aligned}
		& \frac{\partial}{\partial t} c_{1, 0}(t, x_1) + \mathcal{L}_1 c_{1, 0}(t, x_1) = c, \\
		& c_{1, 0}(t, \tilde{\mu}_1^{<}) = 1 - R_1, \quad c_{1, 0}(t, \infty) = -c(T-t), \\
		& c_{1, 0}(T, x_1) = (1 - R_1) \mathbbm{1}_{\{\tilde{\mu}_1^{<} \le x_1 \le \tilde{\mu}_1^{=}\}}, 
\end{aligned}
\end{equation}
and $c_{1,\infty}(t, x_1)$ is the solution of the following boundary value problem
\begin{equation}
\begin{aligned}
		& \frac{\partial}{\partial t} c_{1, \infty}(t, x_1) + \mathcal{L}_1 c_{1, \infty}(t, x_1) = c, \\
		& c_{1, \infty}(t, 0) = 1 - R_1, \quad c_{1, \infty}(t, \infty) = -c(T-t), \\
		& c_{1, \infty}(T, x_1) = (1 - R_1) \mathbbm{1}_{\{x_1 \le \mu_1^{=}\}}.
\end{aligned}
\end{equation}

\subsection{First-to-default swap}
An FTD contract refers to a basket of reference names (RN). Similar to a regular CDS, the Protection Buyer (PB) pays a regular coupon payment $c$ to the Protection Seller (PS) up to the first default of any of the RN in the basket or maturity time $T$. In return, PS compensates PB the loss caused by the first default.

Consider the FTD contract referenced on $2$ banks, and denote its price $F(t, x)$. We assume that the coupon is paid continuously and equals to $c$. Then, the value of FTD contract can be given (\cite{LiptonItkin2015}) by the solution of  (\ref{kolm_1})--(\ref{kolm_2})  with $\chi(t, x) = c$ and terminal condition
\begin{equation*}
	\psi(x) = \beta_0  \mathbbm{1}_{\{x \in D_{12}\}} + \beta_1 \mathbbm{1}_{\{x \in D_{1}\}} + \beta_2 \mathbbm{1}_{\{x \in D_{2}\}},
\end{equation*}
where
\begin{equation*}
	\begin{aligned}
		\beta_0 = 1 - \min[\min(R_1, \tilde{R}_1(\omega_2), \min(R_2, \tilde{R}_2(\omega_1)], \\
		\beta_1 = 1 - \min(R_2, \tilde{R}_2(1)), \quad \beta_2 = 1 - \min(R_1, \tilde{R}_1(1)),
	\end{aligned}
\end{equation*}
and
\begin{equation*}
	\tilde{R}_1(\omega_2) = \min \left[1, \frac{A_1(T) +  \omega_2 L_{2 1}(T)}{L_1(T) + \omega_2 L_{12}(T)}\right], \quad \tilde{R}_2(\omega_1) = \min \left[1, \frac{A_2(T) +  \omega_1 L_{1 2}(T)}{L_2(T) + \omega_1 L_{21}(T)}\right].
\end{equation*}
with $\omega_1 = \omega_1(x)$ and $\omega_2 = \omega_2(x)$ defined in (\ref{term_cond}).

Thus, the pricing problem for a FTD contract is
\begin{equation}
\begin{aligned}
		& \frac{\partial}{\partial t} F(t, x) + \mathcal{L} F(t, x) = c, \\
		& F(t, x_1, 0) = 1 - R_2,  \quad F(t, 0, x_2) = 1 - R_1, \\
		& F(t, x_1, \infty) = f_{2,\infty}(t, x_1), \quad F(t, \infty, x_2) = f_{1,\infty}(t, x_2), \\
		& F(T, x) = \beta_0  \mathbbm{1}_{\{x \in D_{12}\}} + \beta_1 \mathbbm{1}_{\{x \in D_{1}\}} + \beta_2 \mathbbm{1}_{\{x \in D_{2}\}},
\end{aligned}
\end{equation}
where $f_{1,\infty}(t, x_1)$ and $f_{2,\infty}(t, x_2)$ are the solutions of the following boundary value problems
\begin{equation}
\begin{aligned}
		& \frac{\partial}{\partial t} f_{i, \infty}(t, x_i) + \mathcal{L}_i f_{i, \infty}(t, x_i) = c, \\
		& f_{i, \infty}(t, 0) = 1 - R_i, \quad f_{i, \infty}(t, \infty) = -c(T-t), \\
		& f_{1, \infty}(T, x_i) = (1 - R_i) \mathbbm{1}_{\{x_i \le \mu_i^{=}\}}.
\end{aligned}
\end{equation}

\subsection{Credit and Debt Value Adjustments for CDS}

Credit Value Adjustment and Debt Value Adjustment can be considered either unilateral or bilateral. For unilateral counterparty risk, we need to consider only two banks (RN, and PS for CVA and PB for DVA), and a two-dimensional problem can be formulated, while bilateral counterparty risk requires a three-dimensional problem, where Reference Name, Protection Buyer, and Protection Seller are all taken into account. We follow \cite{LiptonSav} for the pricing problem formulation but include jumps and mutual liabilities, which affects the boundary conditions.

\paragraph{Unilateral CVA and DVA}
The Credit Value Adjustment represents the additional price associated with the possibility of a counterparty's default. Then, CVA can be defined as
\begin{equation}
	V^{CVA} = (1- R_{PS}) \mathbb{E}[\mathbbm{1}_{\{\tau^{PS} < \min(T, \tau^{RN}) \}} (V_{\tau^{PS}}^{CDS})^{+} \, | \mathcal{F}_t],
\end{equation}
where $R_{PS}$ is the recovery rate of PS, $\tau^{PS}$ and $\tau^{RN}$ are the default times of PS and RN, and $V_t^{CDS}$ is the price of a CDS without counterparty credit risk.

We associate $x_1$ with the Protection Seller and $x_2$ with the Reference Name, then CVA can be given by the solution of  (\ref{kolm_1})--(\ref{kolm_2})  with $\chi(t, x) = 0$ and $\psi(x) = 0$. Thus,
\begin{equation}
\begin{aligned}
		& \frac{\partial}{\partial t} V^{CVA}+ \mathcal{L} V^{CVA} = 0, \\
		& V^{CVA}(t, 0, x_2) = (1 - R_{PS}) V^{CDS}(t, x_2)^{+}, \quad V^{CVA}(t, x_1, 0) = 0, \\
		& V^{CVA}(T, x_1, x_2) = 0.
\end{aligned}
\end{equation}

Similar, Debt Value Adjustment represents the additional price associated with the default and defined as
\begin{equation}
	V^{DVA} = (1- R_{PB}) \mathbb{E}[\mathbbm{1}_{\{\tau^{PB} < \min(T, \tau^{RN}) \}} (V_{\tau^{PB}}^{CDS})^{-} \, | \mathcal{F}_t],
\end{equation}
where $R_{PB}$ and $\tau^{PB}$ are the recovery rate and default time of the protection buyer.

Here, we associate $x_1$ with the Protection Buyer and $x_2$ with the Reference Name, then, similar to CVA,  DVA can be given by the solution of  (\ref{kolm_1})--(\ref{kolm_2}),
\begin{equation}
\begin{aligned}
		& \frac{\partial}{\partial t} V^{DVA}+ \mathcal{L} V^{DVA} = 0, \\
		& V^{DVA}(t, 0, x_2) = (1 - R_{PB}) V^{CDS}(t, x_2)^{-}, \quad V^{DVA}(t, x_1, 0) = 0, \\
		& V^{DVA}(T, x_1, x_2) = 0.
\end{aligned}
\end{equation}

\paragraph{Bilateral CVA and DVA}

When we defined unilateral CVA and DVA, we assumed that either protection  buyer, or protection seller are risk-free. Here we assume that they are both risky. Then, 
The Credit Value Adjustment represents the additional price associated with the possibility of counterparty's default and defined as
\begin{equation}
	V^{CVA} = (1 - R_{PS}) \mathbb{E}[\mathbbm{1}_{\{\tau^{PS} < \min(\tau^{PB}, \tau^{RN}, T)\}} (V^{CDS}_{\tau^{PS}})^{+} \, | \mathcal{F}_t],
\end{equation} 

Similar, for DVA
\begin{equation}
	V^{DVA} = (1 - R_{PB}) \mathbb{E}[\mathbbm{1}_{\{\tau^{PB} < \min(\tau^{PS}, \tau^{RN}, T)\}} (V^{CDS}_{\tau^{PB}})^{-} \, | \mathcal{F}_t],
\end{equation}

We associate $x_1$ with protection seller, $x_2$ with protection buyer, and $x_3$ with reference name. Here, we have a three-dimensional process. Applying three-dimensional version of (\ref{kolm_1})--(\ref{kolm_2}) with $\psi(x) = 0, \chi(t, x) = 0$, we get
\begin{equation}
	\label{CVA_pde}
\begin{aligned}
		& \frac{\partial}{\partial t} V^{CVA} + \mathcal{L}_3 V^{CVA} = 0, \\
		& V^{CVA}(t, 0, x_2, x_3) = (1 - R_{PS}) V^{CDS}(t, x_3)^{+}, \\
		& V^{CVA}(t, x_1, 0, x_3 ) = 0, \quad V^{CVA}(t, x_1, x_2, 0)  = 0, \\
		& V^{CVA}(T, x_1, x_2, x_3) = 0,
\end{aligned}
\end{equation}
and
\begin{equation}
\label{DVA_pde}
\begin{aligned}
		& \frac{\partial}{\partial t} V^{DVA} + \mathcal{L}_3 V^{DVA} = 0, \\
		& V^{DVA}(t, 0, x_2, x_3) = (1 - R_{PB}) V^{CDS}(t, x_3)^{-}, \\
		& V^{DVA}(t, x_1, 0, x_3 ) = 0, \quad V^{DVA}(t, x_1, x_2, 0)  = 0, \\
		& V^{DVA}(T, x_1, x_2, x_3) = 0,
\end{aligned}
\end{equation}
where $\mathcal{L}_3 f$ is the three-dimensional infinitesimal generator.

\bibliographystyle{apalike}
\bibliography{lit_bib}

\end{document}